\documentclass[11pt]{article}

\newcommand{\showscript}{false}
\newcommand{\version}{long}

\usepackage{algorithm,algpseudocode,amsmath,amsthm,amssymb,bm,complexity,etoolbox,graphicx,ifthen,listings,mathtools,standalone,thmtools,thm-restate}
\ifundef{\tqc}{\usepackage{fullpage}}{}
\usepackage[inline,shortlabels]{enumitem}
\usepackage[section]{placeins}
\usepackage[margin=5pt]{subfig} % Note: subfig is deprecated.

\usepackage{complexity}

\newclass{\ioPSPACE}{i.o.\text{-}PSPACE}
\newlang{\Halt}{Halt}
\ifundef{\tqc}{\usepackage{etoolbox,ifthen,tabulary}

\newcommand{\indraft}[1]{\ifthenelse{\equal{\version}{draft}}{#1}{}}
\newcommand{\infinal}[1]{\ifthenelse{\equal{\version}{final}}{#1}{Only shown in the final version}}
\newcommand{\inshort}[1]{\ifthenelse{\equal{\version}{short}}{#1}{}}
\newcommand{\inlong}[1]{\ifthenelse{\equal{\version}{long}}{#1}{}}
\newcommand{\inshortorlong}[2]{\inshort{#1}\inlong{#2}}

\newcommand{\intqcornot}[2]{\ifdef{\tqc}{#1}{#2}}
\newcommand{\intqc}[1]{\intqcornot{#1}{}}
\newcommand{\notintqc}[1]{\intqcornot{}{#1}}

\newcommand{\appref}[1]{Appendix~\ref{app:#1}}
\newcommand{\secref}[1]{Section~\ref{sec:#1}}

\newcommand{\renameenv}[2]{
  \expandafter\let\csname #1#2\expandafter\endcsname
  \csname #1\endcsname
  \expandafter\let\csname end#1#2\expandafter\endcsname
  \csname end#1\endcsname
  \expandafter\let\csname #2\endcsname\relax
  \expandafter\let\csname end#2\endcsname\relax}

\ifundef{\defaultlists}{
  \usepackage[inline,shortlabels]{enumitem}
  \setenumerate[1]{(a),itemsep=0pt,topsep=3pt,parsep=0pt,partopsep=0pt}
  \setenumerate[2]{(i),noitemsep,topsep=3pt,parsep=0pt,partopsep=0pt}
  \setenumerate[3]{(A),noitemsep,topsep=3pt,parsep=0pt,partopsep=0pt}
  \setenumerate[4]{(I),noitemsep,topsep=3pt,parsep=0pt,partopsep=0pt}
  \setitemize{noitemsep,topsep=3pt,parsep=0pt,partopsep=0pt}
  \setdescription{noitemsep,topsep=3pt,parsep=0pt,partopsep=0pt}
  \setlist{noitemsep,topsep=3pt,parsep=0pt,partopsep=0pt}}{}

\newcolumntype{x}[1]{>{\centering\arraybackslash}m{#1}}
}{} %$
\usepackage{algpseudocode,amsmath,amssymb,etoolbox,fixltx2e,mathtools}

\let\ket\relax
\let\eqref\relax

\DeclareMathOperator{\tr}{tr}

\DeclareFontFamily{U}{mathx}{\hyphenchar\font45}
\DeclareFontShape{U}{mathx}{m}{n}{
      <5> <6> <7> <8> <9> <10>
      <10.95> <12> <14.4> <17.28> <20.74> <24.88>
      mathx10
      }{}
\DeclareSymbolFont{mathx}{U}{mathx}{m}{n}
\DeclareMathSymbol{\bigtimes}{1}{mathx}{"91}

\algnewcommand{\Input}{\item[\textbf{Input:}]}
\algnewcommand{\Output}{\item[\textbf{Output:}]}
\MakeRobust{\Call} % Fix a problem with nested calls.

\newcommand\restr[2]{{
  \left.\kern-\nulldelimiterspace
  #1
  \vphantom{\big|}
  \right|_{#2}
  }}

\newcommand{\abs}[1]{\left\vert#1\right\vert}

\newcommand{\coleq}{\coloneqq}

\newcommand{\norm}[1]{\left\Vert#1\right\Vert}

\newcommand{\trnorm}[1]{\norm{#1}_{\rm{tr}}}

\newcommand{\ket}[1]{\left|#1\right\rangle}

\newcommand{\setb}[2]{\left\{#1 \;\middle|\; #2\right\}}
\newcommand{\kD}[1]{\ensuremath{#1\mathrm{D}}}
\newcommand{\nth}[1]{\ensuremath{{#1}^{\mathrm{th}}}}
\renewcommand{\algref}[1]{Algorithm~\ref{alg:#1}}
\newcommand{\figref}[1]{Figure~\ref{fig:#1}}

\newcommand{\thmref}[1]{Theorem~\ref{thm:#1}}
\newcommand{\lemref}[1]{Lemma~\ref{lem:#1}}

\newcommand{\defref}[1]{Definition~\ref{defn:#1}}

\newcommand{\corref}[1]{Corollary~\ref{cor:#1}}

\newcommand{\eqref}[1]{(\ref{eq:#1})}

\newcommand{\eps}{\epsilon}

\newcommand{\bbN}{\mathbb{N}}

\newcommand{\bbZ}{\mathbb{Z}}

\newcommand{\cD}{\mathcal{D}}
\newcommand{\cE}{\mathcal{E}}
\newcommand{\cF}{\mathcal{F}}

\newcommand{\ra}{\rightarrow}

\newcommand{\CU}[1]{\mathrm{C#1}}

\newcommand{\swap}{\text{swap}}
\newcommand{\teleport}{\text{teleport}}

\newcommand{\B}{\{0, 1\}}

\inshort{\usepackage{amsthm}
\makeatletter
\def\thm@space@setup{\thm@preskip=3pt \thm@postskip=3pt}
\makeatother

\makeatletter
\renewenvironment{proof}[1][\proofname]{\par
  \pushQED{\qed}
  \normalfont
  \topsep3pt \partopsep0pt % No space before
  \trivlist
  \item[\hskip\labelsep
        \itshape
    #1\@addpunct{.}]\ignorespaces
  }{
    \popQED\endtrivlist\@endpefalse
    \addvspace{0pt plus 0pt} % No space after
  }
\makeatother
} %$
\usepackage{etoolbox,ifthen,thmtools,thm-restate}

\ifundef{\dontnumberwithin}{\declaretheorem[numberwithin=section]{dummy}}{\declaretheorem{dummy}} % Using dummy resolves some incompatibility issues with certain documentclasses (e.g. 
\declaretheorem[sibling=dummy]{theorem}

\declaretheorem[sibling=dummy]{lemma}

\declaretheorem[sibling=dummy]{definition}

\declaretheorem[sibling=dummy]{corollary}

\ifundef{\defaultthmcontinues}{\renewcommand{\thmcontinues}[1]{}}{}
\inshort{}
\notintqc{\inshort{\usepackage[compact]{titlesec}}}

\newcommand{\reorder}{\text{Reorder}}
\newcommand{\interact}{\text{Interact}}
\newcommand{\control}{\text{Control}}
\newcommand{\cstage}{\text{Control-Stage}}
\newcommand{\conclock}{\text{Control-Clockwise}}
\newcommand{\crotate}{\text{Rotate}}

\newcommand{\resetsubfig}{\setcounter{subfigure}{0}}

\begin{document}

\title{Optimal Quantum Circuits for Nearest-Neighbor Architectures}
\notintqc{\author{David J. Rosenbaum \\ University of Washington, \\ Department of Computer Science \& Engineering \\ Email: djr@cs.washington.edu}}
\date{May  8, 2013}
\maketitle
\thispagestyle{empty}

\begin{abstract}
  We show that the depth of quantum circuits in the realistic architecture where a classical controller determines which local interactions to apply on the \kD{k} grid $\bbZ^k$ where $k \geq 2$ is the same (up to a constant factor) as in the standard model where arbitrary interactions are allowed.  This allows minimum-depth circuits (up to a constant factor) for the nearest-neighbor architecture to be obtained from minimum-depth circuits in the standard abstract model.  Our work therefore justifies the standard assumption that interactions can be performed between arbitrary pairs of qubits.  In particular, our results imply that Shor's algorithm, controlled operations and fanouts can be implemented in constant depth, polynomial size and polynomial width in this architecture.

  We also present optimal non-adaptive quantum circuits for controlled operations and fanouts on a \kD{k} grid.  These circuits have depth $\Theta(\sqrt[k]{n})$, size $\Theta(n)$ and width $\Theta(n)$.  Our lower bound also applies to a more general class of operations.
\end{abstract}

\notintqc{\newpage}
\notintqc{\setcounter{page}{1}}

\newcommand{\twoDNTC}{
  \begin{figure}[H] % Due to a bug (http://latex-community.org/forum/viewtopic.php?f=45&t=19752), it is necessary to avoid using [H] with \subfloat as this sometimes results in incorrect subfigure numbering.
    \centering
    \subfloat[][Interactions in the \kD{2} NTC architecture: the grid lines indicate the two-qubit interactions which can be performed]{
      \label{fig:2D-NTC-grid}
      \includegraphics{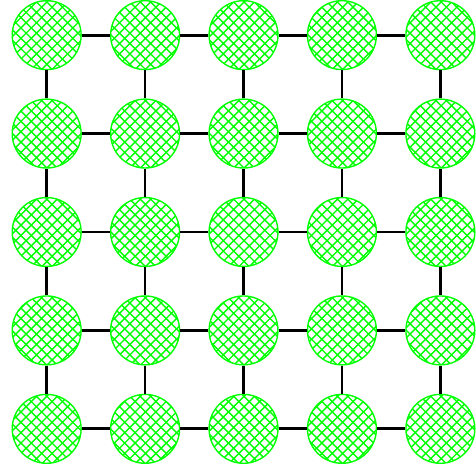}}
    \subfloat[][An example of concurrent interactions in the \kD{2} NTC architecture: the components connected by the thick red edges indicate concurrent interactions and the thick red circles indicate single-qubit interactions]{
      \label{fig:2D-NTC-concurrent}
      \includegraphics{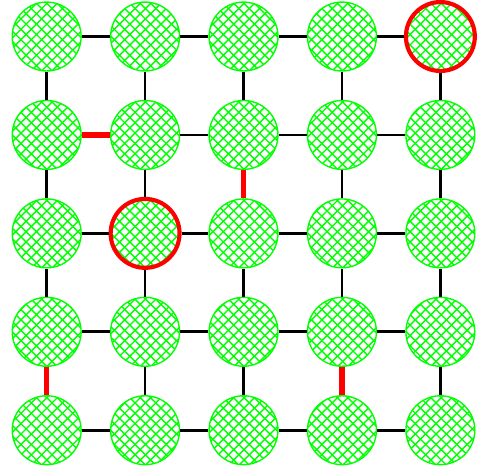}}
    \caption{The \kD{2} NTC architecture}
    \label{fig:2D-NTC}
  \end{figure}
  \resetsubfig
}

\section{Introduction}
Quantum algorithms are typically formulated at an abstract level and allow arbitrary one- and two-qubit interactions.  However, in physical implementations of quantum computers, typically only local interactions between neighboring qubits are possible.  This motivates the \emph{\kD{k} nearest-neighbor two-qubit concurrent} (\kD{k} NTC) architecture~\cite{vanmeter2005a} (cf.~\cite{choi2010a}) in which the qubits are arranged on the \kD{k} grid $\bbZ^k$\notintqc{; this is shown in \figref{2D-NTC-grid} for the case where $k = 2$}.  Operations may involve one or two qubits with the restriction that two-qubit operations may only be performed along an edge in the grid.  Multiple operations may be performed concurrently as long as they are on disjoint sets of qubits\notintqc{; an example is shown in \figref{2D-NTC-concurrent}}. %  The operations applied in the \kD{k} NTC model depend only on the input size $n$ and not on intermediate measurement outcomes.  Thus, \kD{k} NTC circuits are non-adaptive.

\notintqc{\twoDNTC}

The idea of using a classical controller to determine which operations to apply at each step is implicit in the pre- and post-processing stages of Shor's algorithm~\cite{shor1994a} and is often assumed for fault-tolerant quantum computation.  Since the classical controller can take intermediate measurement outcomes into account, this model includes the class of adaptive quantum circuits as a special case.  It is potentially even more powerful since the classical controller can perform randomized polynomial-time computations to determine which operations to apply as well as perform pre- and post-processing.  Since quantum operations are far more expensive than classical operations, we are primarily concerned with the depth of the quantum circuit and do not count the operations performed by the classical controller as long as they take polynomial time.

In this work, we study both the \emph{classical-controller \kD{k} NTC} (\kD{k} CCNTC) architecture --- a classical controller model where interactions are restricted to a \kD{k} grid --- as well as the \emph{non-adaptive \kD{k} NTC}\footnote{The original NTC architecture described by Van Meter and Itoh~\cite{vanmeter2005a} is in fact NANTC; however, we prefer NANTC to avoid confusion with CCNTC where a classical controller is used.} (NANTC) architecture where no classical controller is used and the operations applied cannot depend on intermediate measurement outcomes.  The CCNTC model ignores the cost of offline computations performed by the classical controller and assumes that there are no classical locality restrictions.  This is realistic since the clock rate for a classical computer is much faster than for a quantum computer.  Because quantum computers are already forced to be parallel devices in order to perform operations fault tolerantly~\cite{aharnov1996a}, the total runtime of a quantum circuit is proportional to the depth of the corresponding quantum circuit.  The restriction that interactions are between neighbors on a \kD{k} grid comes from the underlying physical device: in most technologies, only qubits that are spatially close can interact.

We first show how to simulate the standard \emph{classical controller abstract concurrent} (CCAC) architecture in \kD{k} CCNTC with constant factor overhead in the depth.  We accomplish this using a \kD{2} CCNTC teleportation scheme that allows arbitrary interactions on disjoint sets of qubits to be performed in constant depth.

%We first compare the standard \emph{classical controller abstract concurrent}\footnote{This is the AC architecture of Van Meter and Itoh~\cite{vanmeter2005a} augmented with a classical controller.  Note that the AC architecture is not to be confused with the complexity class \AC.} (CCAC) architecture to \kD{k} CCNTC.  Our goal is to simulate CCAC in \kD{k} CCNTC with constant factor overhead in the depth.  We accomplish this using a \kD{2} CCNTC teleportation scheme that allows arbitrary interactions on disjoint sets of qubits to be performed in constant depth.

\begin{restatable}{theorem}{simccntc}
  \label{thm:sim-ccntc}
  Suppose that $C$ is a CCAC quantum circuit with depth $d$, size $s$ and width $n$.  Then $C$ can be simulated in $O(d)$ depth, $O(s n)$ size and $n^2$ width in \kD{2} CCNTC.
\end{restatable}

This result justifies the standard assumption that non-local interactions can be performed efficiently.  Simulating each of the $d$ timesteps from the CCAC circuit in \kD{2} CCNTC requires an $O(n)$ time classical computation; this can be reduced to $O(\log n)$ time if the classical controller is a parallel device or if it includes a simple classical circuit.  Since the clock speeds of classical devices are currently much faster than those of quantum devices, this overhead is not likely to be significant.

\begin{corollary}
  \label{cor:ntc-depth}
  Let $\cE$ be a quantum operation on $n$ qubits.  Let $d_1$ and $d_2$ be the minimum depths\footnote{Here, we assume that there is a minimum depth required to implement $\cE$ in CCAC when the size and width are $\poly(n)$.} required to implement $\cE$ with error at most $\eps$ using $\poly(n)$ size and $\poly(n)$ width in the CCAC and \kD{k} CCNTC models respectively where $k \geq 2$.  Then $d_1 = \Theta(d_2)$.
\end{corollary}

It is possible to implement Shor's algorithm~\cite{shor1994a} in constant depth in CCAC~\cite{browne2010a} which implies that it can also be implemented in constant depth in \kD{2} CCNTC.

\begin{restatable}{corollary}{shorconst}
  \label{cor:shorconst}
  Shor's algorithm can be implemented in constant depth, polynomial size and polynomial width in \kD{2} CCNTC.
\end{restatable}

Since controlled-$U$ operations and fanouts can also be performed in constant depth and polynomial width in CCAC~\cite{hoyer2005a,browne2010a,takahashi2011a}, we also have the following corollary.

\begin{restatable}{corollary}{conUfanout}
  \label{cor:conU-fanout}
  Controlled-$U$ operations with $n$ controls and fanouts with $n$ targets can be implemented in constant depth, $\poly(n)$ size and $\poly(n)$ width in \kD{2} CCNTC.
\end{restatable}

Our main technical result allows any subset of qubits to be reordered in constant depth.  \thmref{sim-ccntc} follows from this as a corollary.

\begin{restatable}{theorem}{inject}
  \label{thm:inject}
  Suppose we have an $n \times n$ grid where all qubits except those in the first column are in the state $\ket{0}$.  Let $T \subseteq \{0, \ldots, n - 1\}$ and let $\pi : T \ra \{0, \ldots, n - 1\}$ be an injection such that for all $j \in T$ with $\pi(j) = 0$, $\setb{k \in T^c}{k < j} = \emptyset$.  Set $m = \abs{\setb{j \in T}{\pi(j) \not= 0}}$. Then we can move each qubit at $(0, j)$ to $(\pi(j), 0)$ for all $j \in T$ in $O(1)$ depth, $O(m n)$ size and $(m + 1) n \leq n^2$ width in \kD{2} CCNTC.
\end{restatable}

Upper bounds for the depth of quantum circuits when converting between various architectures with no classical controller were previously studied by Cheung, Maslov and Severini~\cite{cheung2007a}.  Their results imply that CCAC can be simulated in \kD{k} CCNTC with $O(\sqrt[k]{n})$ factor depth overhead, $O(n)$ size overhead and no width overhead.  In contrast to our results, their techniques are based on applying swap gates to move the interacting qubits next to each other and do not perform any measurements.

Implementations of Shor's algorithm in \kD{k} CCNTC with various super-constant depths were previously known for $k = 1$ and $k = 2$.  Fowler, Devitt and Hollenberg~\cite{fowler2004a} showed a \kD{1} CCNTC circuit for Shor's algorithm which requires $O(n^3)$ depth, $O(n^4)$ size and $O(n)$ width where $n$ is the number of bits in the integer which is being factored.  Maslov~\cite{maslov2007a} showed that any stabilizer circuit can be implemented in linear depth in \kD{1} CCNTC from which the result of Fowler, Devitt and Hollenberg~\cite{fowler2004a} can be recovered.  Kutin~\cite{kutin2006a} gave a more efficient \kD{1} CCNTC circuit which uses $O(n^2)$ depth, $O(n^3)$ size and $O(n)$ width.  For \kD{2} CCNTC, Pham and Svore~\cite{pham2012a} showed an implementation of Shor's algorithm in polylogarithmic depth, polynomial size and polynomial width.

It was also previously known that controlled-$U$ operations and fanouts can be implemented in constant depth, polynomial size and polynomial width in CCAC.  This line of work was started by Moore~\cite{moore1999a} who showed that parity and fanout are equivalent and posed the question of whether fanout has constant-depth circuits.  H{\o}yer and {\v S}palek~\cite{hoyer2005a} proved that if fanout has constant-depth circuits then controlled-$U$ operations can also be implemented in constant depth with inverse polynomial error.  Browne, Kashefi and Predrix~\cite{browne2010a} showed that one-way quantum computation is equivalent to unitary quantum circuits with fanout.  A consequence of this is that constant depth adaptive circuits for fanout can be used to implement controlled-$U$ operations in constant depth in CCAC.  Takahashi and Tani~\cite{takahashi2011a} reduced the size of this circuit by a polynomial and made it exact.

In many technologies, measurements are much more costly than unitary operations.  For this reason, we also consider the non-adaptive \kD{k} NANTC model.  Here, there is no classical controller and the operations applied depend only on the size of the input and not on intermediate measurement outcomes.  Our result in this model is a characterization of the complexity of controlled-$U$ operations and fanouts.

\begin{restatable}{theorem}{nantcopt}
  \label{thm:nantc-opt}
  The depth required for controlled-$U$ operations with $n$ controls and fanouts with $n$ targets in \kD{k} NANTC is $\Theta(\sqrt[k]{n})$.  Moreover, this depth can be achieved with size $\Theta(n)$ and width $\Theta(n)$.
\end{restatable}

If the clock speeds of the quantum computer and its classical controller are comparable, then operations implemented using \thmref{nantc-opt} are significantly faster than those implemented using \corref{conU-fanout}.  For this reason, \thmref{nantc-opt} may become a better option as quantum computing technology matures.

The layout of our paper is as follows.  In \secref{def}, we discuss definitions used in the rest of the paper and define the models of computation precisely.  In \secref{teleport}, we review quantum teleportation and describe teleportation chains.  In \secref{depth-ccntc}, we describe our \kD{2} teleportation scheme and show that it allows arbitrary interactions to be implemented in constant depth in \kD{2} CCNTC.  In \secref{control}, we show an algorithm that implements controlled-$U$ operations and fanouts for \kD{k} NANTC in depth $O(\sqrt[k]{n})$.  In \secref{fanout}, we describe how our techniques can be applied to obtain \kD{k} NANTC quantum circuits for fanout with depth $O(\sqrt[k]{n})$.  In \secref{optimal}, we prove a matching lower bound for a class of operations that includes controlled-$U$ operations and fanouts.

\section{Definitions}
\label{sec:def}
The one- and two-qubit operations that can be performed by the hardware are called the \emph{basic operations}.  We assume that the basic operations are a \emph{universal gate set} so that any one- or two-qubit unitary can be constructed from the basic operations.  We also assume that the basic operations include measurement in the computational basis.

It is useful to distinguish between physical and logical timesteps.  During each \emph{physical timestep}, we can perform any set of disjoint basic operations.  During a \emph{logical timestep}, we allow any set of disjoint $t$-qubit operations to be performed.  In this work, we take $t = O(k)$ and assume $k$ is constant.

\begin{definition}[NANTC]
  \label{defn:nantc}
  In the \kD{k} NANTC model, computation is performed by applying a sequence of sets of basic operations $S_1, \ldots, S_d$ to the \kD{k} grid of qubits.  We require that the operations in the set $S_i$ are disjoint and are either single-qubit operations or two-qubit operations between neighbors in the \kD{k} grid.  The sequence of sets of operations must be randomized polynomial-time computable from the size $n$ of the input.
\end{definition}

In the models where a classical controller is present, the classical controller is invoked after each physical timestep to determine which operations to apply at the next step.

\begin{definition}[CCAC]
  \label{defn:ccac}
  Let $M$ be a randomized polynomial-time machine that takes the input $x$ and the measurement outcomes from the first $i$ physical timesteps and outputs a set $M_1, \ldots, M_\ell$ of disjoint basic operations to be applied to the qubits at the \nth{i + 1} physical timestep.  If no more physical timesteps are to be performed, then $M$ outputs the special symbol $\boxdot$.  Computation in the CCAC model is performed at physical timestep $i$ by using $M$ to compute the set of operations to apply and then applying them to the qubits.
\end{definition}

The CCNTC model is similar except that it also requires that two-qubit operations are only performed between neighbors on the \kD{k} grid.

\begin{definition}[CCNTC]
  \label{defn:ccntc}
  Let $M$ be a randomized polynomial-time machine that takes the input $x$ and the measurement outcomes from the first $i$ physical timesteps and outputs a set $M_1, \ldots, M_\ell$ of disjoint basic operations to be applied to the \kD{k} grid of qubits at the \nth{i + 1} physical timestep.  We require that each $M_i$ is either a single-qubit operation or a two-qubit operation between neighbors in the \kD{k} grid.  If no more physical timesteps are to be performed, then $M$ outputs the special symbol $\boxdot$.  Computation in the CCNTC model is performed at physical timestep $i$ by using $M$ to compute the set of operations to apply and then applying them to the \kD{k} grid of qubits.
\end{definition}

In this paper, the machine $M$ from Definitions~\ref{defn:ccac} and \ref{defn:ccntc} will be deterministic except for the pre- and post-processing stages of Shor's algorithm.

For NANTC, a \emph{quantum circuit} is the sequence of basic operations $M_1, \ldots, M_\ell$ be applied to the \kD{k} grid of qubits.  For the CCAC and CCNTC models, a \emph{quantum circuit} is described by the machine $M$ from Definitions~\ref{defn:ccac} and \ref{defn:ccntc}.  We now define three standard measures of cost in these models.

\begin{definition}
  The depth of a quantum circuit is

  \begin{enumerate}
  \item $d$ for NANTC where $S_1, \ldots, S_d$ is the sequence of operations from \defref{nantc} for an input of size $n$
  \item $\max_{x \in \B^n} \max_r d_{x, r}$ for CCAC and CCNTC where $d_{x, r}$ is the number of physical timesteps it takes for the machine $M$ from Definitions~\ref{defn:ccac} and \ref{defn:ccntc} to output $\boxdot$ when the input is $x$ and the random seed is $r$.  The first max is taken is over all possible inputs $x$ of length $n$ and the second is over all possible random seeds $r$.
  \end{enumerate}
\end{definition}

We note that the depth only changes by a constant factor if we use logical timesteps instead of physical timesteps in the above definition.  This is due to our assumption that any operation performed in a logical timestep acts on at most $O(k) = O(1)$ qubits.

\begin{definition}
  The size of a quantum circuit is
  \begin{enumerate}
  \item $\sum_i \abs{S_i}$ for NANTC where $S_1, \ldots, S_d$ is the sequence of operations from \defref{nantc} for an input of size $n$
  \item $\max_{x \in \B^n} \max_r s_{x, r}$ for CCAC and CCNTC where $S_{x, r}$ is the total number of operations applied when the input is $x$ and the random seed is $r$.  The first max is taken over all possible inputs $x$ of length $n$ and the second is over all possible random seeds $r$.
  \end{enumerate}
\end{definition}

In the next definition, we assume that the qubits are indexed by $\bbN$ for CCAC.

\begin{definition}
  The width of a quantum circuit is
  \begin{enumerate}
  \item the total number of qubits acted on by operations in the sets $S_i$ for NANTC where $S_1, \ldots, S_d$ is the sequence of operations from \defref{nantc} for an input of size $n$
  \item $\max_{x \in \B^n} \abs{A_x}$ for CCAC where $A_x$ is the smallest subset of $\bbN$  such that every qubit acted on is contained in $A_x$ for input $x$ and all random seeds $r$
  \item $\max_{x \in \B^n} \abs{A_x}$ for CCNTC where $A_x$ is the smallest hypercube in $\bbZ^k$ such that every qubit acted on is contained in $A_x$ for input $x$ and all random seeds $r$
  \end{enumerate}
\end{definition}

Typically, the depth is the most important metric to optimize since it is proportional to the amount of time required to execute the quantum operations.  The width is also important since the number of qubits is currently quite limited but the size is largely irrelevant.  Moreover, if parallelism is properly exploited then we expect the size to be roughly the depth times the width.

\section{Quantum teleportation}
\label{sec:teleport}
In this section we review quantum teleportation~\cite{bennett1993a}.  As we shall see, teleportation is a useful primitive that allows non-local interactions to be performed in a constant-depth circuit in \kD{k} CCNTC.  Let us denote the states of the Bell basis by $\ket{\Phi_0} = \frac{\ket{00} + \ket{11}}{\sqrt{2}}$, $\ket{\Phi_1} = \frac{\ket{01} + \ket{10}}{\sqrt{2}}$, $\ket{\Phi_2} = \frac{\ket{01} - \ket{10}}{\sqrt{2}}$ and $\ket{\Phi_3} = \frac{\ket{00} - \ket{11}}{\sqrt{2}}$.  Up to global phase, these can be written as $\ket{\Phi_\ell}^{AB} = \sigma_\ell^B \ket{\Phi_0}^{AB}$.  Recall that in the quantum teleportation setting, Alice has a state $\ket{\psi}^S = \alpha \ket{0}^S + \beta \ket{1}^S$ that she wishes to send to Bob.  The two parties are not allowed to send quantum states to each other but each have one qubit of a Bell state $\sigma_\ell^B \ket{\Phi_0}$ and can communicate classically.

To perform quantum teleportation, Alice performs a Bell measurement on the $SA$ registers.  If the measurement outcome is $\ket{\Phi_k}$, then a simple calculation shows that the resulting state is $\ket{\Phi_k}^{SA} \otimes \sigma_\ell \sigma_k \ket{\psi}^B$.  Alice then sends the classical measurement outcome $k$ to Bob; by applying the appropriate Pauli operation to his register $B$, Bob causes to overall state to become $\ket{\Phi_k}^{SA} \otimes \ket{\psi}^B$.  Observe that Alice's state $\ket{\psi}$ has been recovered in Bob's register.

Let us now consider how quantum teleportation chains can be used in the \kD{1} CCNTC model to perform non-local operations in constant depth.  Suppose that we have a qubit in the state $\ket{\psi}^S$ along with $m$ Bell states $\ket{\Phi_{\ell_j}}^{A_j B_j}$.  These are arranged on a line so that the overall state is $\ket{\psi}^S \bigotimes_{j = 1}^m \ket{\Phi_{\ell_j}}^{A_j B_j}$.  Our goal is to move qubit $S$ to $B_m$.  One way to do this is to first teleport $S$ to $B_1$ by performing a Bell measurement on $S A_1$.  We then store the measurement outcome $k_1$ but do not apply the correcting Pauli operation; at this point, the state of $B_1$ is $\sigma_{\ell_1} \sigma_{k_1} \ket{\psi}$.  Continuing this process, we obtain the state $\bigotimes_{j = 1}^m \ket{\Phi_{k_j}} \prod_{j=m}^1 \left(\sigma_{\ell_j} \sigma_{k_j}\right) \ket{\psi}^{B_m}$.  Since $\prod_{j=m}^1 \left(\sigma_{\ell_j} \sigma_{k_j}\right)$ is just a Pauli operation, we obtain the state $\bigotimes_{j = 1}^m \ket{\Phi_{k_j}} \ket{\psi}^{B_m}$ in a single quantum operation.  The crucial point here is that all of the Bell measurements are performed on disjoint pairs of qubits so they can all be done in parallel as in one-way quantum computation~\cite{raussendorf2001a,raussendorf2002a} and~\cite{terhal2002a}.  Thus, we can perform a non-local interaction of arbitrary distance in constant depth.  It is important to note that this is not possible without a classical controller since otherwise there is no way to compute the correcting Pauli operation.

\newcommand{\aiwccsevenzero}{
  \begin{figure}[H]
    \centering
    \noindent\makebox[\textwidth]{ % Required for proper centering for the TQC version of the paper.
      \subfloat[][]{
        \label{fig:aiwcci-7}
        \includegraphics{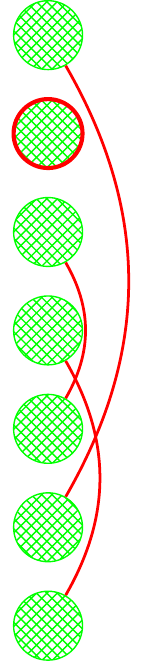}}
      \subfloat[][]{
        \label{fig:aiwcc-7-0}
        \includegraphics{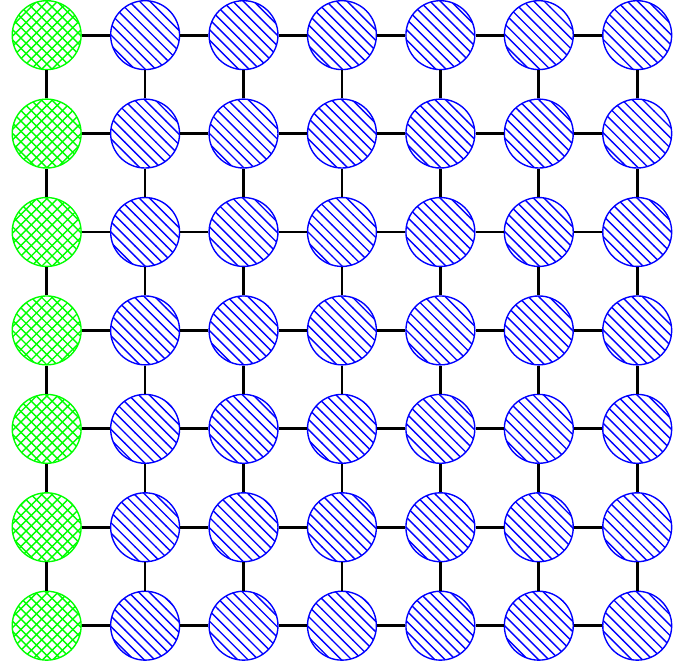}}
      \subfloat[][]{
        \label{fig:aiwcc-7-1}
        \includegraphics{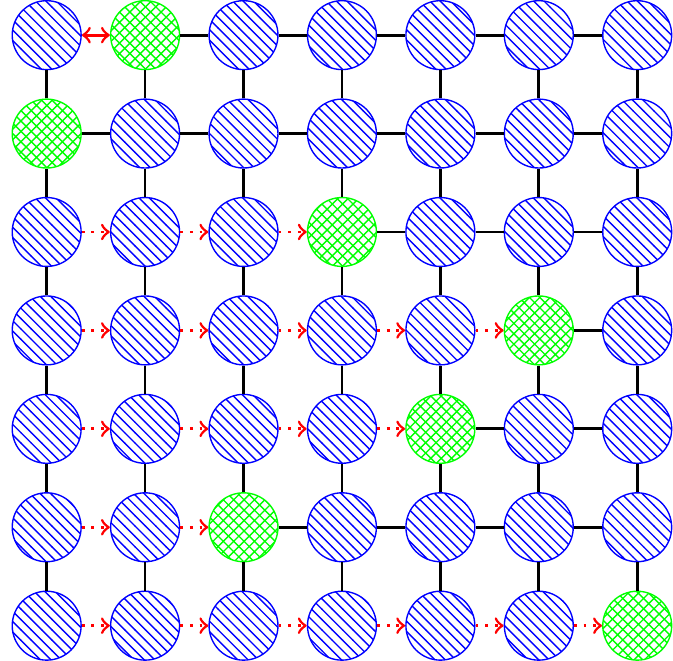}}}
    \caption{Performing an arbitrary set of interactions in \kD{2} CCNTC.  The qubits crosshatched green are the data qubits and the qubits shaded with diagonal downward blue lines are ancilla qubits.}
    \label{fig:aiwcc-7}
  \end{figure}
}

\section{Depth complexity in \kD{k} CCNTC}
\label{sec:depth-ccntc}
In this section, we show that an arbitrary set of CCAC interactions corresponding to basic operations can be performed in constant depth in \kD{2} CCNTC.  We assume that there are $n$ qubits on which the interactions are to be performed and store these in the first column of a \kD{2} $n \times n$ CCNTC grid.  The qubit at location $(i, j)$ is denoted by $q_{i, j}$.  Since we must handle interactions between qubits that are not neighbors, we may as well assume that the original $n$ qubits are stored in the first column $q_{0, 0}, \ldots, q_{0, n - 1}$ of qubits.  The remaining columns are used as ancillas to implement teleportation chains.  We teleport each of the $n$ qubits horizontally to the right so that interacting pairs are in adjacent columns.  Since these teleportations are on disjoint sets of qubits, they can be performed in parallel as in~\cite{raussendorf2001a,raussendorf2002a,terhal2002a}.  A second set of vertical teleportation chains is then used to move all the qubits down to the first row.  At this point, the interacting qubits are neighbors so the interactions may be implemented directly.  We then perform the reverse teleportations to move the qubits back to their original positions.

\newcommand{\aiwccseventwo}{
  \begin{figure}[H]
    \ContinuedFloat
    \centering
    \subfloat[][]{
      \label{fig:aiwcc-7-2}
      \includegraphics{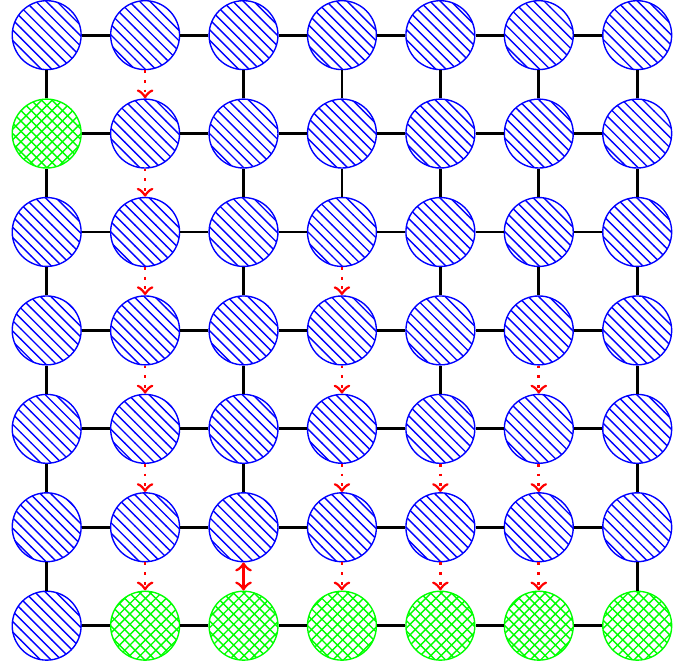}}
    \subfloat[][]{
      \label{fig:aiwcc-7-3}
      \includegraphics{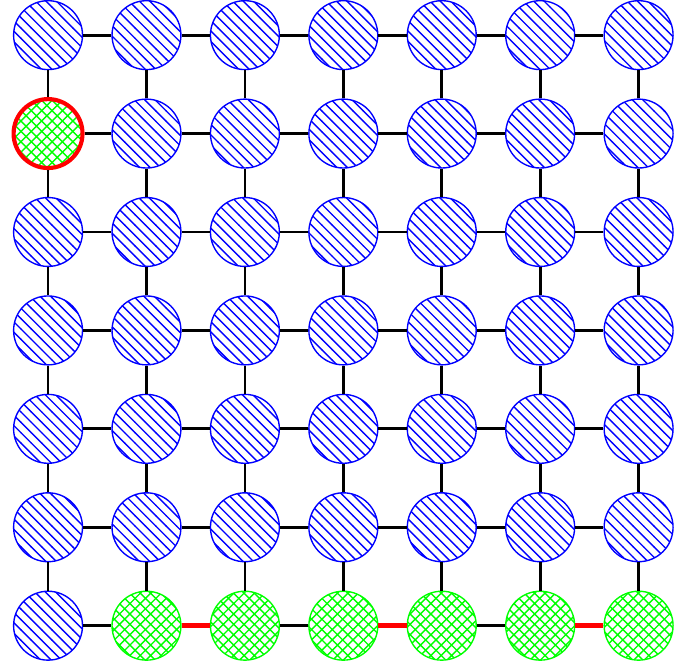}}
    \caption[]{Performing an arbitrary set of interactions in \kD{2} CCNTC}
  \end{figure}
  \resetsubfig
}

\subsection{An example of arbitrary interactions in \kD{2} CCNTC}
We show an example in \figref{aiwcc-7}.  The desired interactions are shown in \figref{aiwcci-7}.  The layout of the data qubits in the \kD{2} grid is shown in \figref{aiwcc-7-0}; the ancilla qubits are used to implement the teleportation chains and are initially set to $\ket{0}$.  We start by horizontally teleporting the qubits that interact to adjacent columns in \figref{aiwcc-7-1} where the teleportation chains are denoted by the dotted red arrows.  The red double arrow indicates a swap operation; this is just a less expensive way of achieving the same result when the qubits are neighbors. The next step is to vertically teleport the data qubits down to the first row as shown in \figref{aiwcc-7-2}.  Finally, all interacting qubits are now neighbors so we perform the desired interactions in \figref{aiwcc-7-3}.  The final reverse teleportations are not shown but can be obtained by reversing the arrows in Figures~\ref{fig:aiwcc-7-1} and \ref{fig:aiwcc-7-2}.

\notintqc{\aiwccsevenzero}
\notintqc{\aiwccseventwo}

\subsection{An algorithm for performing arbitrary interactions in \kD{2} CCNTC}
In order to define our algorithm, we first show how to perform an arbitrary reordering of the positions of the qubits in constant depth.  We assume that there are $n$ data qubits which are located in the first column of the $n \times n$ grid; the remaining qubits are in the state $\ket{0}$.  We let $T \subseteq \{0, \ldots, n - 1\}$ be a subset of row indexes on which an injection $\pi : T \ra \{0, \ldots, n - 1\}$ is to be applied.  This injection describes where the qubits with row indexes in $T$ are to be moved to on the $x$-axis.  The reason we specify $T$ explicitly is because this allows us to only perform teleportations on qubits which have row indexes in $T$.  If $\abs{T} = o(n)$ then this can result in a circuit that has asymptotically smaller size.  The reordering can be applied using \algref{inject} which is based on the same technique as \figref{aiwcc-7}.  The notation $\teleport(q_{i_1, j_1}, q_{i_2, j_2})$ where $i_1 = i_2$ or $j_1 = j_2$ means that a teleportation chain is applied to move the state of qubit at $(i_1, j_1)$ along the line to $(i_2, j_2)$.

\newcommand{\injectalg}{
  \begin{algorithm}[H]
    \centering
    \begin{algorithmic}[1]
       \Require{The $n$ data qubits are in the first column, $T \subseteq \{0, \ldots, n - 1\}$ and $\pi : T \ra \{0, \ldots, n - 1\}$ is an injection.  For all $j \in T$ such that $\pi(j) = 0$, $\setb{k \in T^c}{k < j} = \emptyset$}
       \Ensure{Each qubit at $(0, j)$ is moved to $(\pi(j), 0)$ for all $j \in T$}
       \Function{\reorder}{$T$, $\pi$}
         \For{$j \in T$}
           \State $\teleport(q_{0, j}, q_{\pi(j), j})$
         \EndFor
         \For{$j \in T$}
           \State $\teleport(q_{\pi(j), j}, q_{\pi(j), 0})$
         \EndFor
       \EndFunction
    \end{algorithmic}
    \caption{The algorithm for performing an arbitrary reordering of a subset of the qubits in \kD{2} CCNTC}
    \label{alg:inject}
  \end{algorithm}
}

\notintqc{\injectalg}

\intqc{\aiwccsevenzero}
\intqc{\aiwccseventwo}

Our main technical result follows immediately from \algref{inject}.

\inject*

We note that the $\teleport$ operations in \algref{inject} require an $O(n)$ time classical computation to determine the correcting Pauli matrix (see \secref{teleport}).  Since this computation simply involves multiplying $O(n)$ Pauli matrices, it can be done more efficiently in $O(\log n)$ time by arranging the multiplications in a binary tree.  The $O(\log n)$ runtime requires either that the classical controller is a parallel device or that it includes a special classical circuit for computing the correcting Pauli operation.  Since classical operations are much faster than quantum operations on current devices, this overhead is unlikely to be a problem.

It is now straightforward to describe the algorithm for performing arbitrary interactions.  We first note that an arbitrary set of interactions can be defined by disjoint one and two element subsets $J_k$ of $\{0, \ldots, n - 1\}$ and basic operations $M_k$ where $1 \leq k \leq \ell$ and the values in $J_k$ denote the qubits on which the operation $M_k$ is to be applied.  The pseudocode for performing arbitrary interactions in \kD{2} CCNTC is shown in \algref{aiwcc}.

\newcommand{\aiwcc}{
  \begin{algorithm}[H]
    \centering
    \begin{algorithmic}[1]
      \Require{The $n$ data qubits are in the first column, each $J_k$ is a disjoint one or two element subset of $\{0, \ldots, n - 1\}$ and $M_k$ is a basic operation for $1 \leq k \leq \ell$.  Moreover, $\abs{J_{k_1}} \leq \abs{J_{k_2}}$ for $k_1 \leq k_2$}
      \Ensure{The interactions specified by $J_k$ and $M_k$ are applied}
      \Function{\interact}{$J_1, \ldots, J_\ell$, $M_1, \ldots, M_\ell$}
        \State $T \coleq ()$
        \State $i \coleq 0$
        \For{$k \coleq 1, \ldots, \ell$}
          \If{$\abs{J_k} = 1$}
            \State $i \coleq 1$
          \Else
            \State $\{j_1, j_2\} \coleq J_k$ where $j_1 < j_2$
            \State $\pi(j_1) \coleq i$
            \State $\pi(j_2) \coleq i + 1$
            \State Append the elements of $J_k$ to $T$
            \State $i \coleq i + 2$
          \EndIf
        \EndFor
        \State $\reorder(T, \pi)$
        \State $i \coleq 0$
        \For{$k \coleq 1, \ldots, \ell$}
          \If{$\abs{J_k} = 1$}
            \State $\{j\} \coleq J_k$
            \State Apply $M_k$ to $q_{0, j}$
            \State $i \coleq 1$
          \Else
            \State Apply $M_k$ to $q_{i, 0}, q_{i + 1, 0}$
            \State $i \coleq i + 2$
          \EndIf
        \EndFor
        \State Perform the reverse teleportations to move the qubits back to their original positions
      \EndFunction
    \end{algorithmic}
    \caption{The algorithm for performing arbitrary interactions in \kD{2} CCNTC}
    \label{alg:aiwcc}
  \end{algorithm}
}

\notintqc{\aiwcc}

The following theorem is a direct consequence of \algref{aiwcc}.

\simccntc*

Recalling the discussion following \thmref{inject}, we see that each of the $O(d)$ timesteps requires an $O(n)$ time classical computation if the classical controller is a sequential device or a $O(\log n)$ time computation if it is parallel or includes a simple classical circuit.  The time required to perform a single quantum operation is currently much longer than the time required to execute an instruction on a classical processor so this overhead is likely to be negligible.

The rest of our results for \kD{k} CCNTC follow from \thmref{sim-ccntc}.  Let $\cD_n$ denote the set of all $n \times n$ density matrices.  A general quantum operation is represented as a completely positive trace preserving (CPTP) map $\cE : \cD_n \ra \cD_n$.  Obviously, any circuit in the \kD{2} CCNTC model can also be applied when arbitrary interactions are allowed.  The following corollary is immediate.

\begin{corollary}[continues=cor:ntc-depth]
  Let $\cE : \cD_n \ra \cD_n$ be a CPTP map and let $\eps \geq 0$.  Let $d_1$ and $d_2$ be the minimum depths required to implement $\cE$ with error at most $\eps$ in the CCAC and \kD{k} CCNTC models respectively where $k \geq 2$.  Then $d_1 = \Theta(d_2)$.
\end{corollary}

\intqc{\injectalg}

It is known that Shor's algorithm can be implemented in constant depth, polynomial size and polynomial width in CCAC~\cite{browne2010a} from which we obtain another corollary.

\shorconst*

Because controlled-$U$ operations and fanouts with unbounded numbers of control qubits or targets can be performed in constant depth, polynomial size and polynomial width in CCAC~\cite{hoyer2005a,browne2010a,takahashi2011a}, we have the following result.

\conUfanout*

\intqc{\aiwcc}

\section{Controlled operations in \kD{k} NANTC}
\label{sec:control}
In this section, we show how to control a single-qubit $U$ operation by $n$ controls using $O(\sqrt[k]{n})$ operations in \kD{k} NANTC.  We start with an $m \times m$ grid; for reasons that will become clear later, we require that $m$ is odd.  The control qubits are placed such that they are not at adjacent grid points; the central $3 \times 3$ square has no controls except when $m = 3$.  This is illustrated in Figures~\ref{fig:concc-3x3-0}\intqcornot{ and}{,} \ref{fig:concc-5x5-0}\notintqc{, \ref{fig:concc-7x7-0} and \ref{fig:concc-9x9-0}} for the cases where $m = 3$\intqcornot{ and}{,} $m = 5$\notintqc{, $m = 7$ and $m = 9$}.  Let $c$ be the center of the grid which corresponds to the target qubit.  The circuit works by considering each square ring in the grid with center $c$ (i.e., a set of points in the grid that all have the same distance to the center under the $\ell_\infty$ norm).  We start with the outermost such ring and propagate its control values into the next ring.  At each such step, some of the control values are combined so that all the values can fit into the smaller ring.  This continues until we reach a $3 \times 3$ ring at which point we apply a special sequence of operations to finish applying the controlled operation to the central qubit.  We will show that each stage can be implemented in constant depth so the overall depth is $O(\sqrt{n})$.

\subsection{The base case: the $3 \times 3$ grid}
We now describe how this circuit works in greater detail.  First, consider the case where $m = 3$.  The grid starts as shown in \figref{concc-3x3-0}; note that we do not force the central $3 \times 3$ square to be devoid of controls in this case since this is the entire grid.  All ancilla qubits start in the state $\ket{0}$.  We start by setting the lower left and upper right corner ancilla qubits to the ANDs of their neighboring controls as shown in \figref{concc-3x3-1}.  Both of these operations are disjoint, so this can be done in one logical timestep.  The next step is to swap these two corner qubits with the vertical middle qubits so they can interact with the central target qubit; this is done in \figref{concc-3x3-2}.  Finally, we apply a $U$ operation to the target qubit and control by the two middle qubits in \figref{concc-3x3-3}.

At this point, the target qubit has the desired value; however, there are two other ancilla qubits in \figref{concc-3x3-3} that must have their values uncomputed.  This is done by applying the operations of Figures~\ref{fig:concc-3x3}\subref*{fig:concc-3x3-1}--\subref*{fig:concc-3x3-2} in reverse order.

\newcommand{\conccthree}{
  \begin{figure}[H]
    \centering
    \subfloat[][]{
      \label{fig:concc-3x3-0}
      \includegraphics{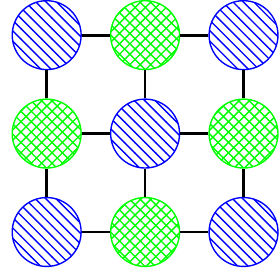}}
    \subfloat[][]{
      \label{fig:concc-3x3-1}
      \includegraphics{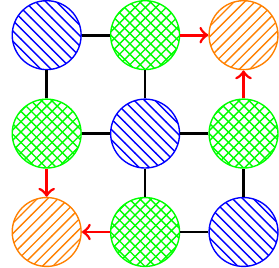}}
    \subfloat[][]{
      \label{fig:concc-3x3-2}
      \includegraphics{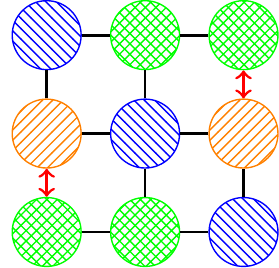}}
    \subfloat[][]{
      \label{fig:concc-3x3-3}
      \includegraphics{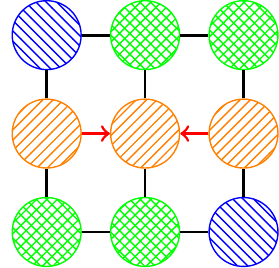}}
    \caption{A controlled operation on a $3 \times 3$ grid.  The qubits crosshatched green are the data qubits, the qubits shaded with diagonal upward orange lines are ancilla qubits which store intermediate data and the qubits shaded with diagonal downward blue lines are ancilla qubits which are currently unused.}
    \label{fig:concc-3x3}
  \end{figure}
  \resetsubfig
}

\notintqc{\conccthree}

\subsection{An example of the general case: the $5 \times 5$ grid}
We now consider an example of the general case where $m = 5$ as shown in \figref{concc-5x5-0}.  The first step is to propagate the values of the outer ring inwards; since the inner ring is $3 \times 3$, there are no controls in the inner ring so this can be done as shown in \figref{concc-5x5-1}.  We then rotate the inner ring as in \figref{concc-5x5-2}.  At this point, the remaining operations to perform are the same as in the $3 \times 3$ case and are shown in Figures~\ref{fig:concc-5x5}\subref*{fig:concc-5x5-3}--\subref*{fig:concc-5x5-5}.  At this point the target qubit has the desired value so we uncompute the intermediate ancillas by applying the operations of Figures~\ref{fig:concc-5x5}\subref*{fig:concc-5x5-1}--\subref*{fig:concc-5x5-4} in reverse order.

\intqc{\conccthree}

\newcommand{\conccfive}{
  \begin{figure}[H]
    \centering
    \subfloat[][]{
      \label{fig:concc-5x5-0}
      \includegraphics{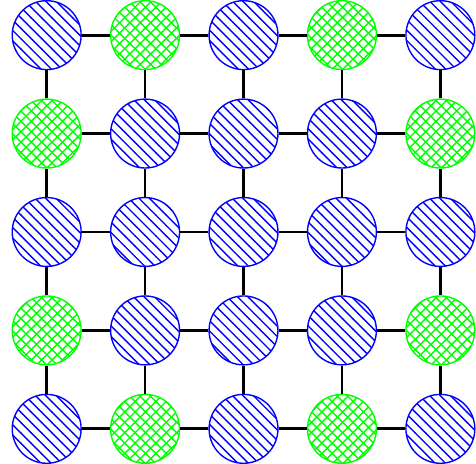}}
    \subfloat[][]{
      \label{fig:concc-5x5-1}
      \includegraphics{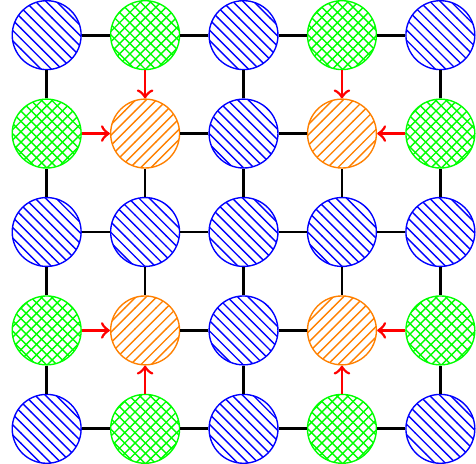}} \\
    \subfloat[][]{
      \label{fig:concc-5x5-2}
      \includegraphics{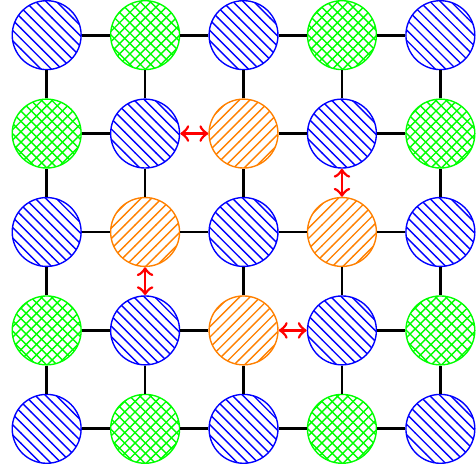}}
    \subfloat[][]{
      \label{fig:concc-5x5-3}
      \includegraphics{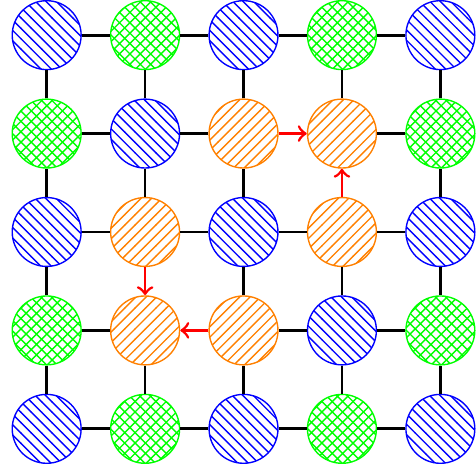}} \\
    \subfloat[][]{
      \label{fig:concc-5x5-4}
      \includegraphics{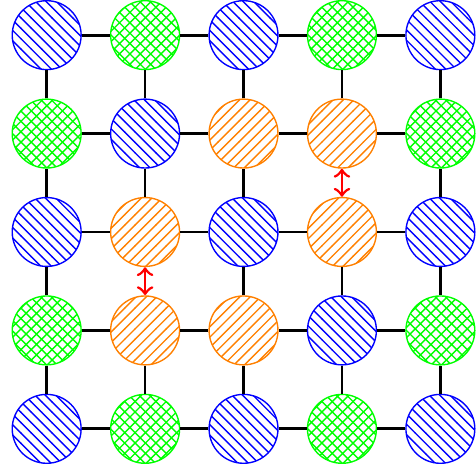}}
    \subfloat[][]{
      \label{fig:concc-5x5-5}
      \includegraphics{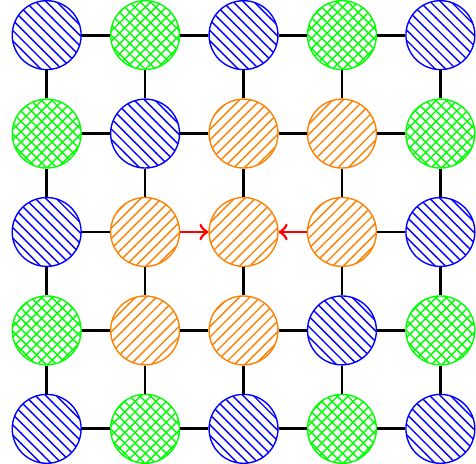}}
    \caption{A controlled operation on a $5 \times 5$ grid.  See \figref{concc-3x3} for the meaning of the colors and shading used.}
    \label{fig:concc-5x5}
  \end{figure}
  \resetsubfig
}

\newcommand{\controltwoDalg}{
  \begin{algorithm}[H]
    \centering
    \begin{algorithmic}[1]
      \Require{$m$ is odd}
      \Ensure{A controlled-$U$ operation is applied to the target}
      \Function{\control}{$m$}
        \State $k \coleq 0$
        \While{$m - 2 k \geq 3$}
          \State $\cstage(k)$
          \State $k \coleq k + 1$
        \EndWhile
        \State Uncompute the intermediate ancillas by repeating all operations except for the final $\CU{U}$ operation in reverse order
      \EndFunction
      \Function{\cstage}{$k$} \Comment{$k$ is the depth of the recursive call}
        \If{$k > 0$}
          \State $\conclock(k)$
          \State $\crotate(k)$
        \EndIf
        \If{$m - 2 k = 3$} \Comment{In this case, we have a $3 \times 3$ grid}
          \State $q_{k,k} \leftarrow q_{k,k} \oplus q_{k,k+1} \wedge q_{k+1,k}$ \label{line:3x3-start}
          \State $q_{k+2,k+2} \leftarrow q_{k+2,k+2} \oplus q_{k+1,k+2} \wedge q_{k+2,k+1}$
          \State $\swap(q_{k,k}, q_{k,k+1})$
          \State $\swap(q_{k+2,k+1}, q_{k+2,k+2})$ \label{line:3x3-controls-done}
          \State $\CU{U}(q_{k+1,k+1}, q_{k,k+1}, q_{k+2,k+1})$ \label{line:3x3-U}
        \EndIf
      \EndFunction
    \end{algorithmic}
    \caption{The algorithm for implementing a controlled-$U$ operation on an $m \times m$ grid}
    \label{alg:control-2D}
  \end{algorithm}
}
 
The same idea applies to an $m \times m$ grid except that when the inner rings have controls (i.e. for $m \geq 7$), the controls from the outer ring must be combined with those in the inner ring at the same time they are propagated inwards.  See \inshortorlong{the full version of our paper~\cite{rosenbaum2012c}}{\appref{examples}} for examples of the $7 \times 7$ and $9 \times 9$ cases.

\intqc{\conccfive}

\subsection{An algorithm for controlled-$U$ operations in $O(\sqrt{n})$ depth in \kD{2} NANTC}
We now present the algorithm used in Figures~\ref{fig:concc-3x3} -- \ref{fig:concc-\intqcornot{5x5}{9x9}} for the general $m \times m$ grid.  Consider an odd $m > 3$.  We denote the coordinates of the qubits on this grid by $(x, y)$ where $0 \leq x, y < m$.  Let $G$ be the set $\{0, \ldots, m - 1\}^2$ of all points on the grid and let $c = ((m - 1) / 2, (m - 1) / 2)$ be the central point.  As discussed previously, the geometry induced by the $\ell_\infty$ norm is useful for reasoning about this grid.  From now on, all distances in this subsection are understood to be with respect to the $\ell_\infty$ norm.

\intqc{\controltwoDalg}

We will say that the \emph{\nth{k} ring} is the set of points that have distance $(m - 1) / 2 - k$ to $c$ so the zeroth ring is outermost; we denote by $R_k = (r_0^k, \ldots, r_{\ell_k}^k)$ the points of the \nth{k} ring where $r_0^k$ is the bottom left corner and the rest of the points are in clockwise order.

The ring $R_k$ contains $4 \left(\frac{m - 1}{2} - k\right)$ controls so the entire grid has $n = 4 \sum_{3 < m - 2 k \leq m} \left(\frac{m - 1}{2} - k\right) = (1 / 2) (m^2 - 9 / 2)$ controls for $m > 3$.  In the case where $m = 3$, there are $4$ controls.  Thus, it is indeed the case that the depth is $O(\sqrt{n})$.

We denote by $q_{i,j}$ the value stored at the point $(i, j)$ and assume the operation to apply to the target is $U$.  The notation $\CU{U}(y, x_1, \ldots, x_\ell)$ denotes applying a controlled-$U$ operation to qubit $y$ conditional on $x_1, \ldots, x_\ell$.  To apply a swap operation to qubits $x$ and $y$, we write $\swap(x, y)$.  The pseudocode for the main algorithm is shown in \algref{control-2D}; the auxiliary functions are shown in \algref{control-aux}.

\notintqc{\conccfive}
\notintqc{\controltwoDalg}

\newcommand{\controlauxalg}{
  \begin{algorithm}[H]
    \centering
    \begin{algorithmic}[1]
      \Function{\conclock}{$k$}
        \State $C = ((k, k), (k, m - k - 1), (m - k - 1, m - k - 1), (m - k - 1, k))$ \Comment{The corners of $R_k$}
        \State $D = ((0, 1), (1, 0), (0, -1), (-1, 0))$ \Comment{The directions to follow between the corners of $R_k$}
        \For{$i \coleq 0, \ldots, 3$}
          \State $i_- \coleq i - 1 \mod 4$
          \State $i_+ \coleq i + 1 \mod 4$
          \State $q_{C_i} \leftarrow q_{C_i} \oplus q_{C_i - D_i} \wedge q_{C_i + D_{i_-}}$ \Comment{Compute the corner ancilla}
          \State Let $s_0, \ldots, s_{\ell_k / 4}$ be the points in $R_k$ from $C_i$ to $C_{i_+}$ excluding $C_{i_+}$
          \State $j \coleq 2$
          \While{$j < \ell_k / 4 - 1$} \Comment{Store the AND of two values in each ancilla in $L$ except for the last}
            \State $q_{L_j} \leftarrow q_{L_j} \oplus q_{L_j - D_i} \wedge q_{L_j + D_{i_-}}$
            \State $j \coleq j + 2$
          \EndWhile
          \State $p \coleq L_{\ell_k / 4 - 1}$
          \If{$m - 2 k > 3$} \Comment{For the last ancilla, use three controls unless we have a $5 \times 5$ grid}
            \State $q_p \leftarrow q_p \wedge q_{p - D_i} \wedge q_{p + D_{i_-}} \wedge q_{p + D_i}$
          \Else
            \State $q_p \leftarrow q_p \wedge q_{p - D_i} \wedge q_{p + D_{i_-}}$
          \EndIf
        \EndFor
      \EndFunction
      \Function{\crotate}{$k$}
        \State $i \coleq 1$
        \While{$i \leq \ell_k$}
          \State $i_+ \coleq i + 1 \mod \ell_k$
          \State $\swap(q_{r_i^k}, q_{r_{i_+}^k})$
          \State $i \coleq i + 2$
        \EndWhile
      \EndFunction
    \end{algorithmic}
    \caption{The ROTATION and CONTROL-CLOCKWISE operations}
    \label{alg:control-aux}
  \end{algorithm}
}

%TQC: Writing out the algorithm for implementing controlled operations in an $m \times m$ grid is relatively straightforward given the examples already shown.  We give the algorithm in \appref{conccalgs}.  It follows immediately that one can implement controlled operations in $O(\sqrt{n})$ depth.

The following theorem is an immediate consequence of \algref{control-2D}.

\begin{theorem}
  \label{thm:2D-control}
  Controlled-$U$ operations with $n$ controls have depth $O(\sqrt{n})$, size $O(n)$ and width $O(n)$ in \kD{2} NANTC.
\end{theorem}

\subsection{Generalization to \kD{k} NANTC}
In this section, we discuss how the circuit can be generalized to $k$ dimensions.  The algorithm works in the same way except the ring $R_k$ is replaced by the grid points on the surface of the hypercube formed by the points at $\ell_\infty$ distance $(m - 1) / 2 - k$ from the center $c$ of the grid.  We proceed as before and propagate the controls on $R_k$ into $R_{k+1}$ until we obtain a grid of width $3$.  Since the number of controls on a \kD{k} grid of length $m$ is $O(m^k)$, we obtain a circuit of depth $O(\sqrt[k]{n})$ for implementing a controlled-$U$ operation with $n$ controls.  The constant depends on $k$, but we assumed that $k$ is constant in \secref{def}.  From this, we obtain the following result.

\begin{theorem}
  \label{thm:kD-control}
  Controlled-$U$ operations with $n$ controls have depth $O(\sqrt[k]{n})$, size $O(n)$ and width $O(n)$ in \kD{k} NANTC.
\end{theorem}

\notintqc{\controlauxalg}

\section{Fanout operations}
\label{sec:fanout}
In this section, we describe quantum circuits for fanout.  In this case, we have a single control qubit and our goal is to XOR it into each of the target qubits.  The construction of fanout circuits is adapted from \algref{control-2D}; the circuits are the same except that the qubit that was the target becomes the control qubit and qubits that were the controls become the targets.  Let $n$ be the number of targets.  In the case of the circuit of \secref{control}, we simply apply all operations in reverse order and replace each Toffoli gate $y \leftarrow y \oplus x_1 \wedge \ldots \wedge x_n$ with a fanout operation $x_j \leftarrow x_j \oplus y$ for all $1 \leq j \leq n$.  This yields a \kD{k} NANTC fanout circuit of depth $O(\sqrt[k]{n})$.  We have shown the following.

\begin{theorem}
  \label{thm:kD-fanout}
  fanouts to $n$ targets have depth $O(\sqrt[k]{n})$, size $O(n)$ and width $O(n)$ in \kD{k} NANTC.
\end{theorem}

\section{Optimality}
\label{sec:optimal}
In this section, we prove that the depth, size and width of the circuits generated by \algref{control-2D} (and its \kD{k} generalization) are optimal for NANTC.  A similar lower bound for addition is discussed in~\cite{choi2011a}.  These lower bounds hold regardless of where the controls and target qubits are located on the \kD{k} grid.  They also hold for a more general class of operations that contains the controlled-$U$ operations and fanouts.

\intqc{\controlauxalg}

Since each qubit is acted on by a constant number of operations in \algref{control-2D}, the size of the circuit is $O(n)$.  This is clearly optimal since any circuit that implements a controlled operation must act on each of the controls.

\begin{theorem}
  \label{thm:CU-lb}
  Any NANTC quantum circuit that implements a non-trivial controlled-$U$ operation with $n$ controls has size $\Omega(n)$.
\end{theorem}

The \emph{trace norm} of a density matrix $\rho$ (denoted $\trnorm{\rho}$) is equal to $(1 / 2) \tr \abs{\rho}$\notintqc{ (the $(1 / 2)$ factor ensures that $\norm{\rho - \sigma}_1$ is the probability of distinguishing $\rho$ and $\sigma$ with the best possible measurement)}.  Consider a general quantum operation $\cE : \cD_n \ra \cD_n$ represented as a CPTP map.  We will use an operator version of the trace norm defined by $\trnorm{\cE} = \sup_{\rho \in \cD} \norm{\cE(\rho)}_1$; if $\cE_1$ and $\cE_2$ are two CPTP maps then $\trnorm{\cE_1 - \cE_2}$ is the probability of distinguishing between them on the worst possible input.  Thus, it is a measure of how much these operations differ.  We will also make use of the partial trace.  If $x$ is a qubit, then we will denote the partial trace over all qubits except $x$ by $\tr_{\neg x} = \tr_{\bbZ^k \setminus \{x\}}$.

Controlled-$U$ operations are special case of a more general class of operations.

\begin{definition}
  \label{defn:input-sens}
  Let $\cE : \cD_n \ra \cD_n$ be a CPTP map.  We say that $\cE$ is $\eps$-input sensitive if there exists a qubit $y$ such that for $\Omega(n)$ qubits $x$, there exists a CPTP map $\cF : \cD_n \ra \cD_n$ acting only on $x$ such that $\trnorm{\tr_{\neg y} (\cE \cF - \cE)} \geq \eps$.
\end{definition}

Intuitively, an $\eps$-input sensitive operation is a generalization of a Toffoli gate where modifying some input qubit $x$ yields a different value on the output with probability $\eps$.  Similarly, we can define $\eps$-output sensitive operations which are generalizations of fanout.

\begin{definition}
  \label{defn:output-sens}
  Let $\cE : \cD_n \ra \cD_n$ be a CPTP map.  We say that $\cE$ is $\eps$-output sensitive if there exists a qubit $x$ such that for $\Omega(n)$ qubits $y$, there exists a CPTP map $\cF : \cD_n \ra \cD_n$ acting only on $x$ such that $\trnorm{\tr_{\neg y} (\cE \cF - \cE)} \geq \eps$.
\end{definition}

We say that $\cE$ is \emph{$\eps$-sensitive} if it is $\eps$-input or $\eps$-output sensitive.  A family $\{\cE : \cD_n \ra \cD_n\}$ of CPTP maps is $\eps$-sensitive if every $\cE_n$ is $\eps$-sensitive.  Our lower bounds will apply to all families of $\eps$-sensitive operations.  All proofs will be for the case of $\eps$-input sensitive operations but the argument of $\eps$-output sensitive operations is all but identical.

\begin{theorem}
  \label{thm:sens-size-lb}
  Let $\{\cE_n : \cD_n \ra \cD_n\}$ be a family of $\eps$-sensitive operations.  Then any family of \kD{k} NANTC circuits $\{C_n\}$ such that $\trnorm{\cE_n - C_n} < \eps / 2$ for all $n$ has size $\Omega(n)$.
\end{theorem}

\begin{proof}
  Suppose that $C_n$ has size $o(n)$.  Assume $\cE_n$ is $\eps$-input sensitive and choose a qubit $y$ as in definition \defref{input-sens} (the case where it is $\eps$-output sensitive is very similar).  There are $\Omega(n)$ qubits $x$ such that there exists a CPTP map $\cF : \cD_n \ra \cD_n$ acting only on $x$ such that $\trnorm{\tr_{\neg y} (\cE_n \cF - \cE_n)} \geq \eps$.  For large $n$, there is such an $x$ which is not acted on by $C_n$.  Then $\tr_{\neg y} C_n \cF = \tr_{\neg y} C_n$.  Now \intqc{$\trnorm{\tr_{\neg y} (C_n - \cE_n)} \geq \abs{\trnorm{\tr_{\neg y} (C_n \cF - \cE_n \cF)} - \trnorm{\tr_{\neg y} (\cE_n \cF - \cE_n)}} > \eps / 2$}
\notintqc{
  \begin{align}
    \trnorm{\tr_{\neg y} (C_n - \cE_n)} & = \trnorm{\tr_{\neg y} (C_n \cF - \cE_n)} \\
                                     {} & \geq \abs{\trnorm{\tr_{\neg y} (C_n \cF - \cE_n \cF)} - \trnorm{\tr_{\neg y} (\cE_n \cF - \cE_n)}} \\
                                     {} & > \eps / 2
  \end{align}

}
which is a contradiction.
\end{proof}

We call a controlled-$U$ operation \emph{non-trivial} if $U \not= I$.  It is easy to prove the following.

\begin{lemma}
  \label{lem:CU-fanout-sens}
  Non-trivial controlled-$U$ operations and fanouts are $1$-sensitive.
\end{lemma}

We now obtain a corollary of \thmref{sens-size-lb} of which \thmref{CU-lb} is a special case.

\begin{corollary}
  \label{cor:CU-fanout-size-lb}
  Let $\{\cE_n : \cD_n \ra \cD_n\}$ denote a family of controlled-$U$ operations or fanouts.  Any family of \kD{k} NANTC circuits $\{C_n\}$ such that $\trnorm{C_n - \cE_n} < 1 / 2$ has size $\Omega(n)$.
\end{corollary}

This shows that \algref{control-2D} (and its \kD{k} generalization) have optimal size.  Next, we will show that $\eps$-sensitive \kD{k} NTC circuits have depth $\Omega(\sqrt[k]{n})$.  For this we require the following easy lemma.

\begin{lemma}
  \label{lem:packing}
  For any subset $S \subseteq \bbZ^k$ and any $x \in \bbZ^k$, there exists a subset $T \subseteq S$ of size $\Omega(\abs{S})$ such that for all $y \in T$, $\norm{x - y}_1 = \Omega(\sqrt[k]{\abs{S}})$.
\end{lemma}

We are now ready to prove our depth lower bound.

\begin{theorem}
  \label{thm:sens-depth-lb}
  Let $\{\cE_n : \cD_n \ra \cD_n\}$ be a family of $\eps$-sensitive operations.  Then any family of \kD{k} NANTC circuits $\{C_n\}$ such that $\trnorm{\cE_n - C_n} < \eps / 2$ for all $n$ has depth $\Omega(\sqrt[k]{n})$.
\end{theorem}

\begin{proof}
  Suppose $\{C_n\}$ has depth $t = o(\sqrt[k]{n})$.  Assume that $\cE_n$ is $\eps$-input sensitive (the case where it is $\eps$-output sensitive is very similar) and choose a qubit $y$ as in \defref{input-sens}.  There is a set $S$ of $\Omega(n)$ qubits such that for each $x \in S$, there exists a CPTP map $\cF : \cD_n \ra \cD_n$ acting only on $x$ with $\trnorm{\tr_{\neg y} (\cE_n \cF - \cE_n)} \geq \eps$.  Let $c > 0$ be the hidden constant in the expression $\Omega(\sqrt[k]{\abs{S}})$ from \lemref{packing}.  For sufficiently large $n$, the depth of $C_n$ is strictly less than $c \sqrt[k]{n}$.  Let $G_i$ be the set of disjoint one- and two-qubit operations that are performed at timestep $1 \leq i \leq t$ in $C_n$.  For an operation $M \in G_i$, let us say that $M$ is \emph{active} if

  \begin{enumerate}
  \item $M$ acts non-trivially on $y$ or
  \item there is an operation $M' \in G_j$ with $i < j \leq t$ such that $M'$ is active and $M$ and $M'$ act non-trivially on a common qubit
  \end{enumerate}

  Let us say that a qubit $x$ \emph{influences} $y$ if there exists an active operation $M \in G_i$ that acts non-trivially on $x$.  Suppose $x$ influences $y$ after $t$ timesteps.  Because all operations act on pairs of adjacent qubits, the $\ell_1$ distance between $x$ and $y$ is at most $t$.  By \lemref{packing}, there exists a subset $T$ of $S$ of size $\Omega(n)$ such that $\norm{x - y}_1 \geq c \sqrt[k]{n}$ for all $x \in T$.  Because $t < c \sqrt[k]{n}$, $x$ does not influence $y$ for $x \in T$.  Let us fix some $x \in T$.  Choosing a $\cF$ acting only on $x$ as in \defref{input-sens}, we have \intqc{$\trnorm{\tr_{\neg y} (C_n - \cE_n)} \geq \abs{\trnorm{\tr_{\neg y} (C_n \cF - \cE_n \cF)} - \trnorm{\tr_{\neg y} (\cE_n \cF - \cE_n)}} > \eps / 2$}
\notintqc{
  \begin{align}
    \trnorm{\tr_{\neg y} (C_n - \cE_n)} & = \trnorm{\tr_{\neg y} (\cF C_n - \cE_n)} \\
                                     {} & \geq \abs{\trnorm{\tr_{\neg y} (C_n \cF - \cE_n \cF)} - \trnorm{\tr_{\neg y} (\cE_n \cF - \cE_n)}} \\
                                     {} & > \eps / 2
  \end{align}

}
  which is a contradiction.
\end{proof}

By \lemref{CU-fanout-sens}, we obtain the following corollary.

\begin{corollary}
  \label{cor:CU-fanout-depth-lb}
  Let $\{\cE_n : \cD_n \ra \cD_n\}$ denote a family of controlled-$U$ operations or fanouts.  Any family of \kD{k} NANTC circuits $\{C_n\}$ such that $\trnorm{C_n - \cE_n} < 1 / 2$ has depth $\Omega(\sqrt[k]{n})$.
\end{corollary}

From Theorems~\ref{thm:kD-control} and \ref{thm:kD-fanout} and Corollaries~\ref{cor:CU-fanout-size-lb} and~\ref{cor:CU-fanout-depth-lb}, we conclude that \algref{control-2D} and its \kD{k} generalization are optimal in their depth, size and width.

\nantcopt*

% \section{Conclusion}
% In this work, we saw that quantum teleportation can be used to implement quantum circuits with arbitrary interactions in a \kD{2} architecture where only operations between neighboring qubits are allowed with only a constant factor increase in the depth.  However, this comes at the cost of a quadratic increase in the size of the quantum circuit.

\intqc{\vspace{-15pt}}
\section*{Acknowledgments}
\intqc{\vspace{-6pt}}

I thank Paul Beame and Aram Harrow for useful discussions and feedback and the anonymous reviewers for helpful comments.  Aram Harrow suggested the use of teleportation chains as a primitive.  Paul Pham suggested applying the technique of \algref{control-2D} to fanouts.  I was funded by the DoD AFOSR through an NDSEG fellowship.  Partial support was provided by IARPA under the ORAQL project.

\newcommand{\references}{
  \intqc{\vspace{-12pt}}
  \bibliographystyle{initials}
  \bibliography{$HOME/LaTeX/computer-science-references,$HOME/LaTeX/math-references,$HOME/LaTeX/quantum-computing-references} %$
}

\intqc{\references}

\appendix
\ifthenelse{\equal{\showscript}{true}}{
  \section{Generating the figures}
  The figures in this paper were formatted using tikz code that was generated using the following Python 3 script:
  
  \lstinputlisting[breaklines = true, language = Python]{draw.py}}{}

\inlong{\newpage
  \FloatBarrier
  \section{More Examples}
  \label{app:examples}
  We now present the implementation of controlled-$U$ operations in $7 \times 7$ and $9 \times 9$ \kD{2} NANTC grids.  This is shown for $m = 7$ in \figref{concc-7x7}.  As before, it is necessary to uncompute the intermediate ancillas by applying the operations of Figures~\ref{fig:concc-7x7}\subref*{fig:concc-7x7-1}--\subref*{fig:concc-7x7-6} in reverse order.  We also show the case where $m = 9$ in \figref{concc-9x9}.  In this case, we apply the operations of Figures~\ref{fig:concc-9x9}\subref*{fig:concc-9x9-1}--\subref*{fig:concc-9x9-8} in reverse order to uncompute the intermediate ancillas.

  \begin{figure}[H]
    \centering
    \subfloat[][]{
      \label{fig:concc-7x7-0}
      \includegraphics{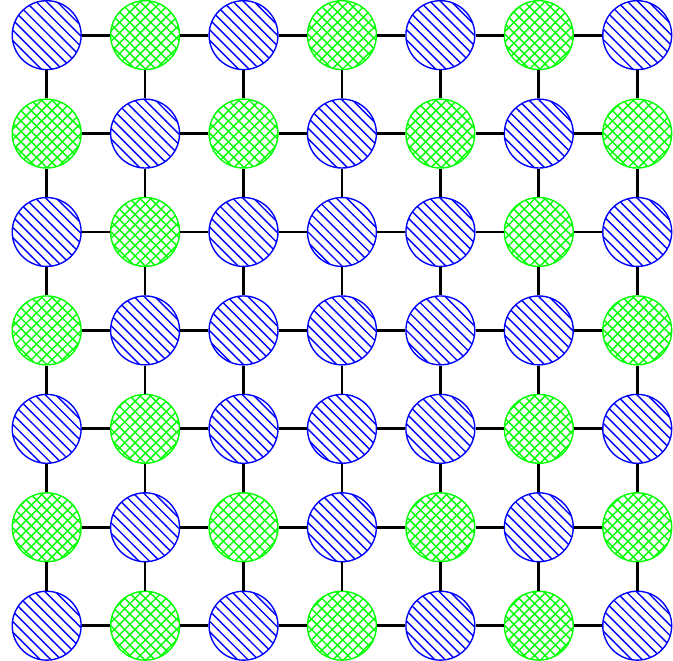}}
    \subfloat[][]{
      \label{fig:concc-7x7-1}
      \includegraphics{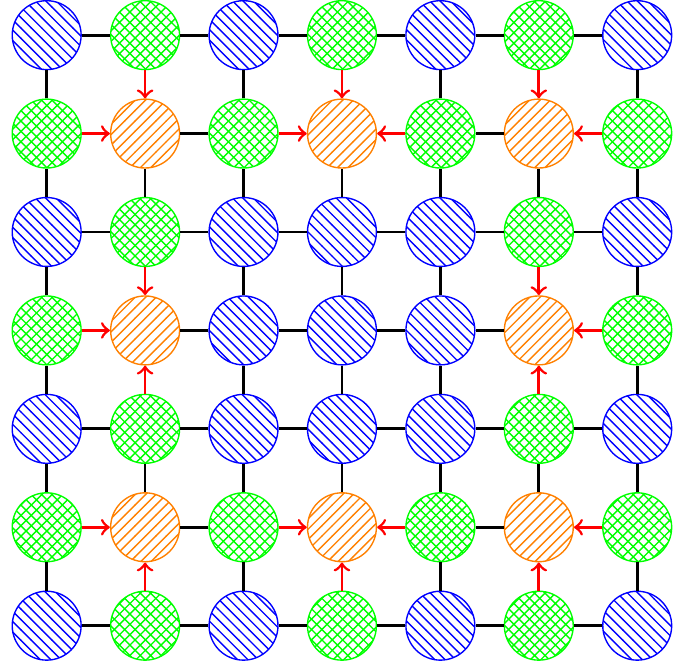}} \\
    \subfloat[][]{
      \label{fig:concc-7x7-2}
      \includegraphics{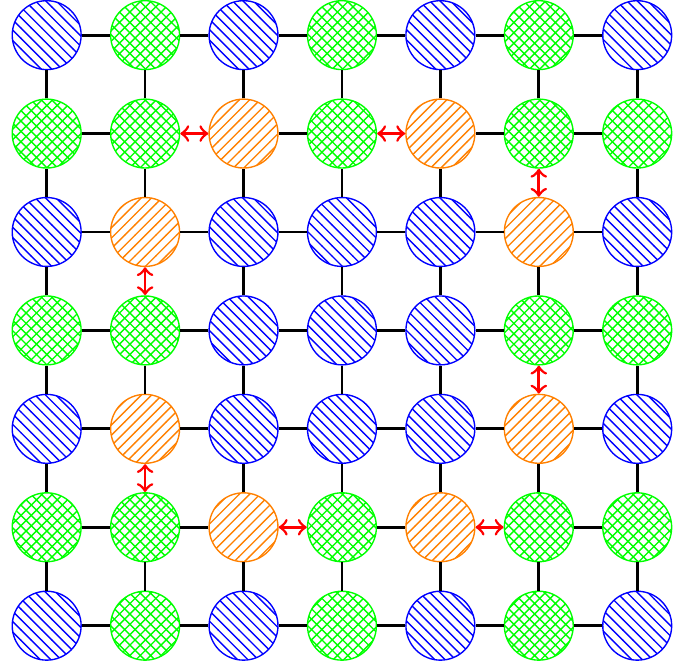}}
    \subfloat[][]{
      \label{fig:concc-7x7-3}
      \includegraphics{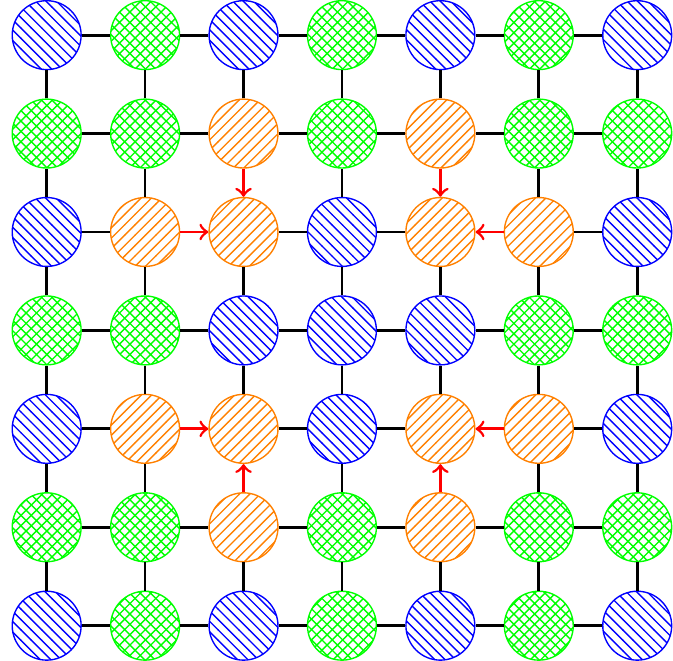}}
    \caption{A controlled operation on a $7 \times 7$ grid.  See \figref{concc-3x3} for the meaning of the colors and shadings used.}
    \label{fig:concc-7x7}
  \end{figure}

  \begin{figure}[H]
    \ContinuedFloat
    \centering
    \subfloat[][]{
      \label{fig:concc-7x7-4}
      \includegraphics{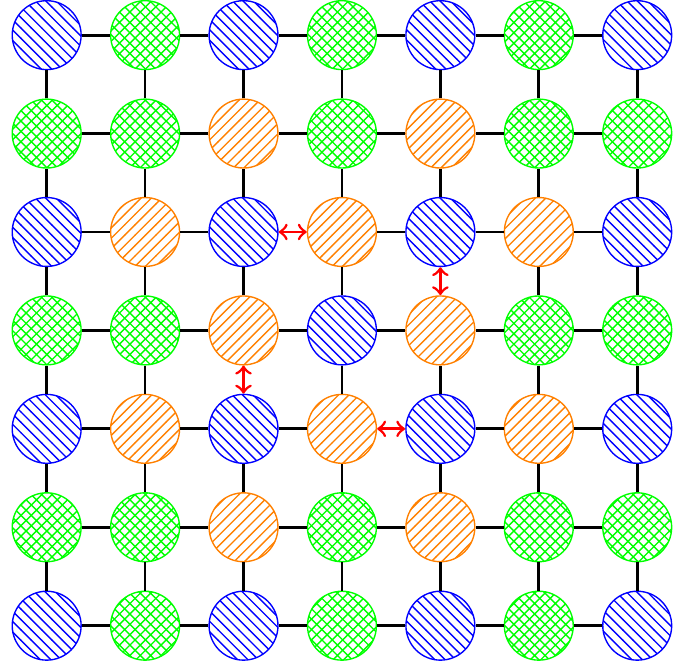}}
    \subfloat[][]{
      \label{fig:concc-7x7-5}
      \includegraphics{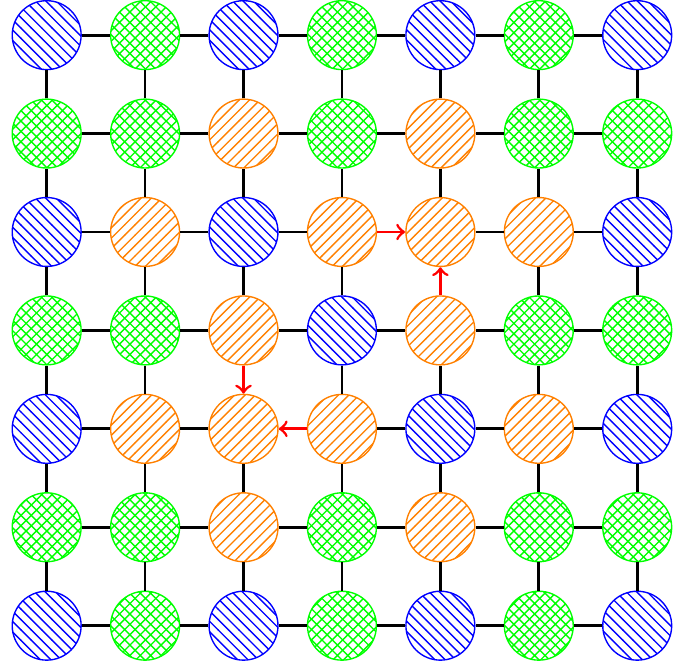}} \\
    \subfloat[][]{
      \label{fig:concc-7x7-6}
      \includegraphics{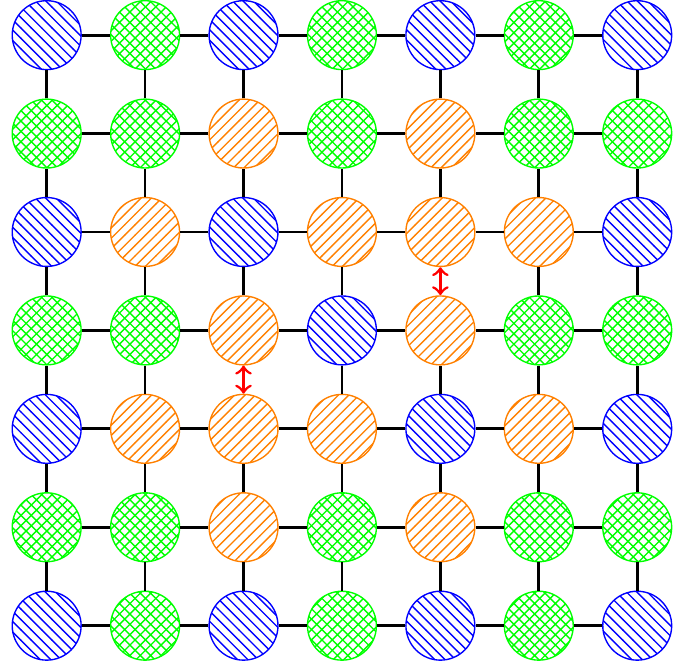}}
    \subfloat[][]{
      \label{fig:concc-7x7-7}
      \includegraphics{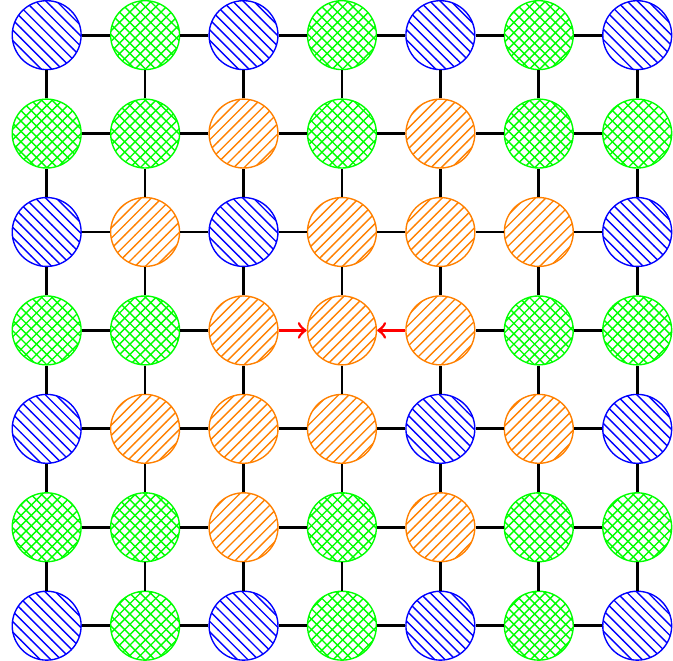}}
    \caption[]{A controlled operation on a $7 \times 7$ grid}
  \end{figure}

  \resetsubfig

  \begin{figure}[H]
    \centering
    \noindent\makebox[\textwidth]{
      \subfloat[][]{
        \label{fig:concc-9x9-0}
        \includegraphics{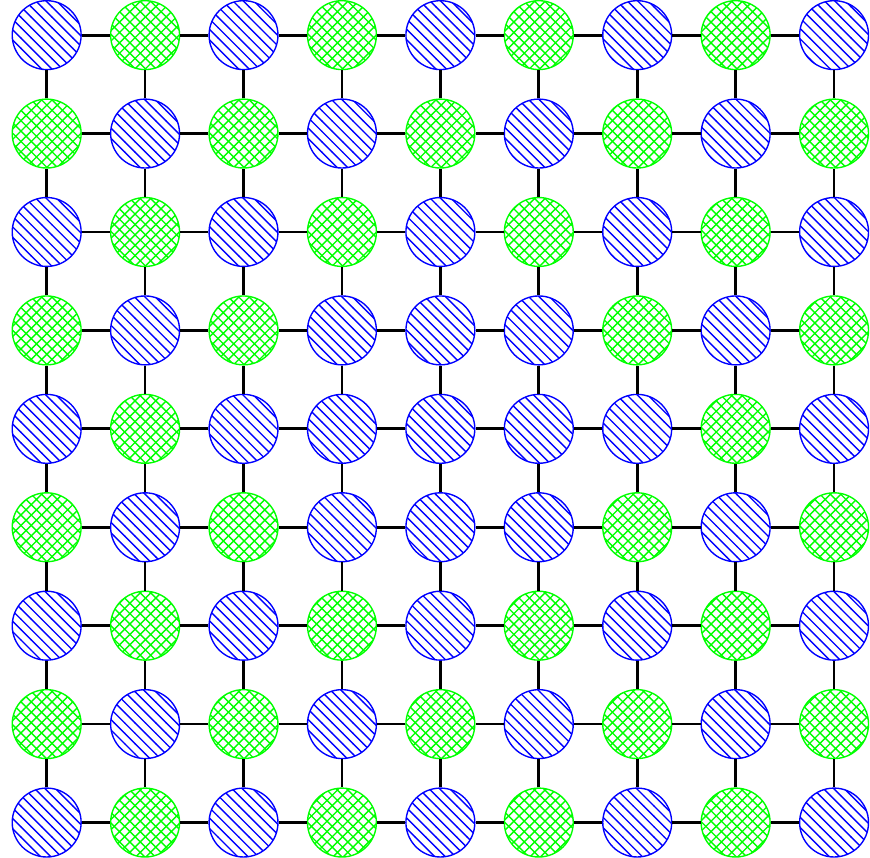}}
      \subfloat[][]{
        \label{fig:concc-9x9-1}
        \includegraphics{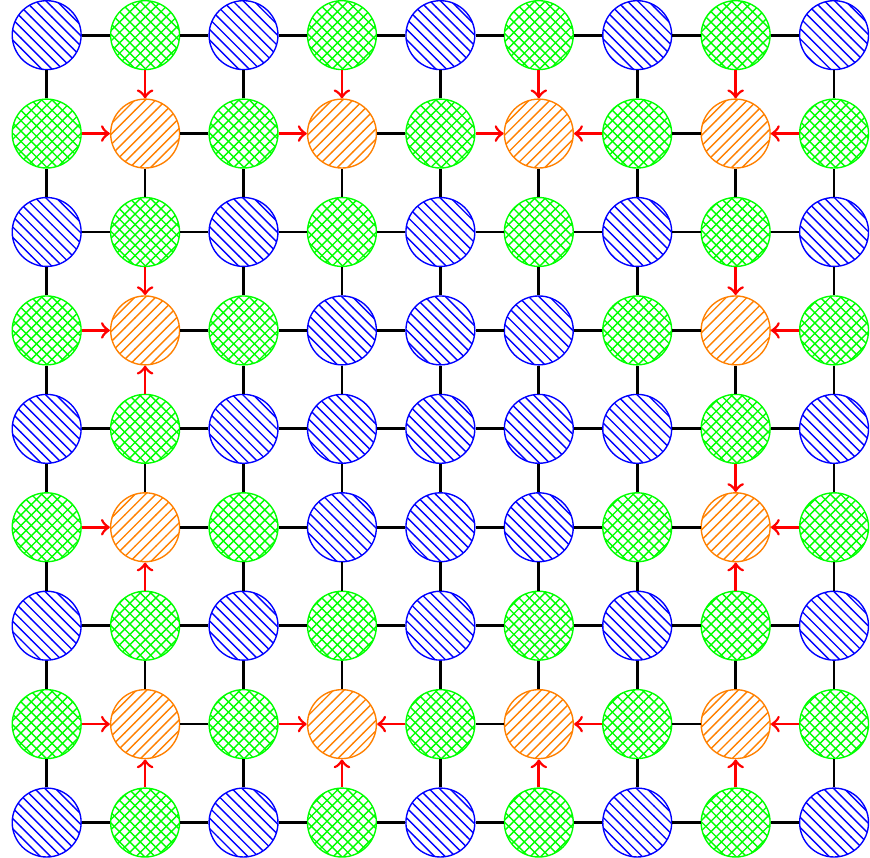}}} \\
    \noindent\makebox[\textwidth]{
      \subfloat[][]{
        \label{fig:concc-9x9-2}
        \includegraphics{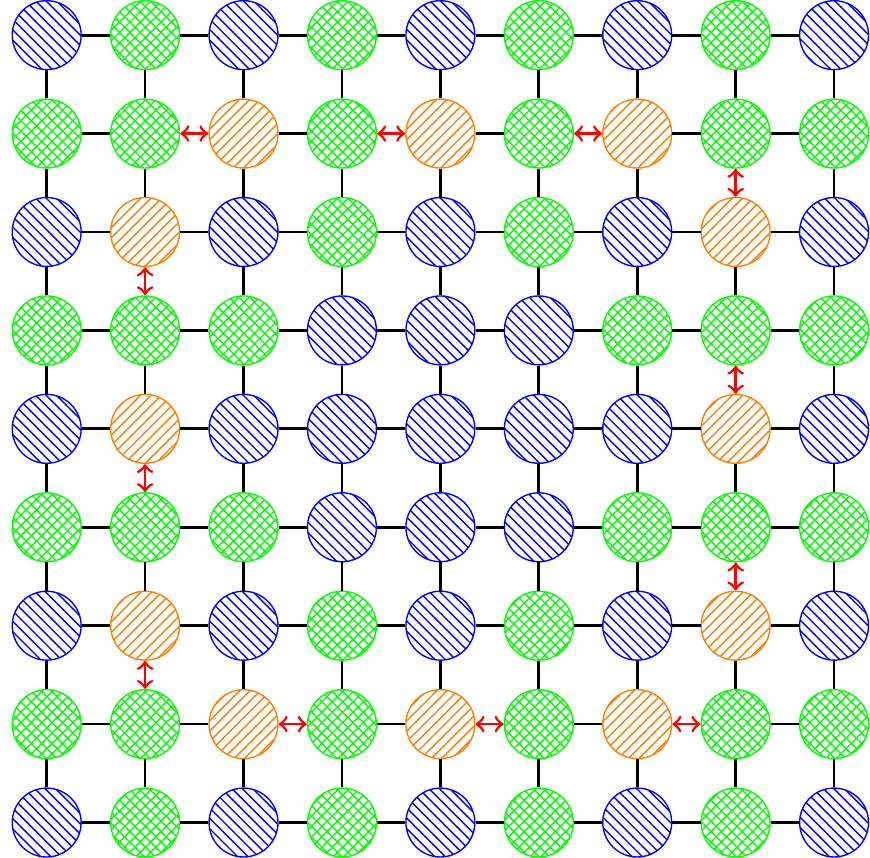}}
      \subfloat[][]{
        \label{fig:concc-9x9-3}
        \includegraphics{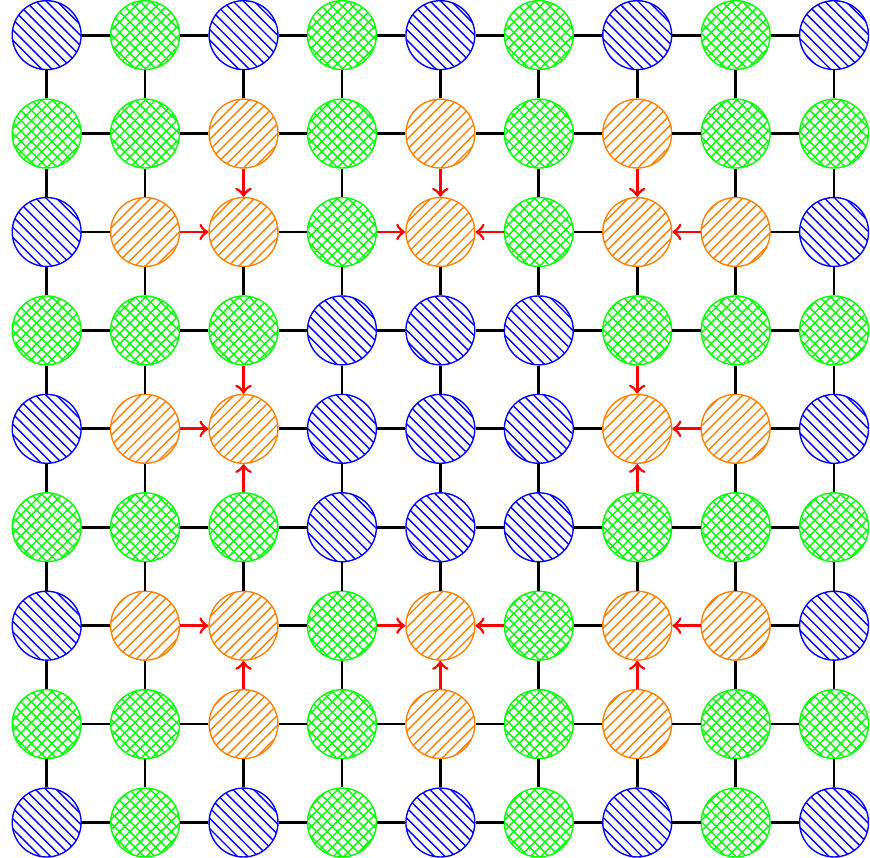}}}
    \caption{A controlled operation on a $9 \times 9$ grid.  See \figref{concc-3x3} for the meaning of the colors and shadings used.}
    \label{fig:concc-9x9}
  \end{figure}

  \begin{figure}[H]
    \ContinuedFloat
    \centering
    \noindent\makebox[\textwidth]{
      \subfloat[][]{
        \label{fig:concc-9x9-4}
        \includegraphics{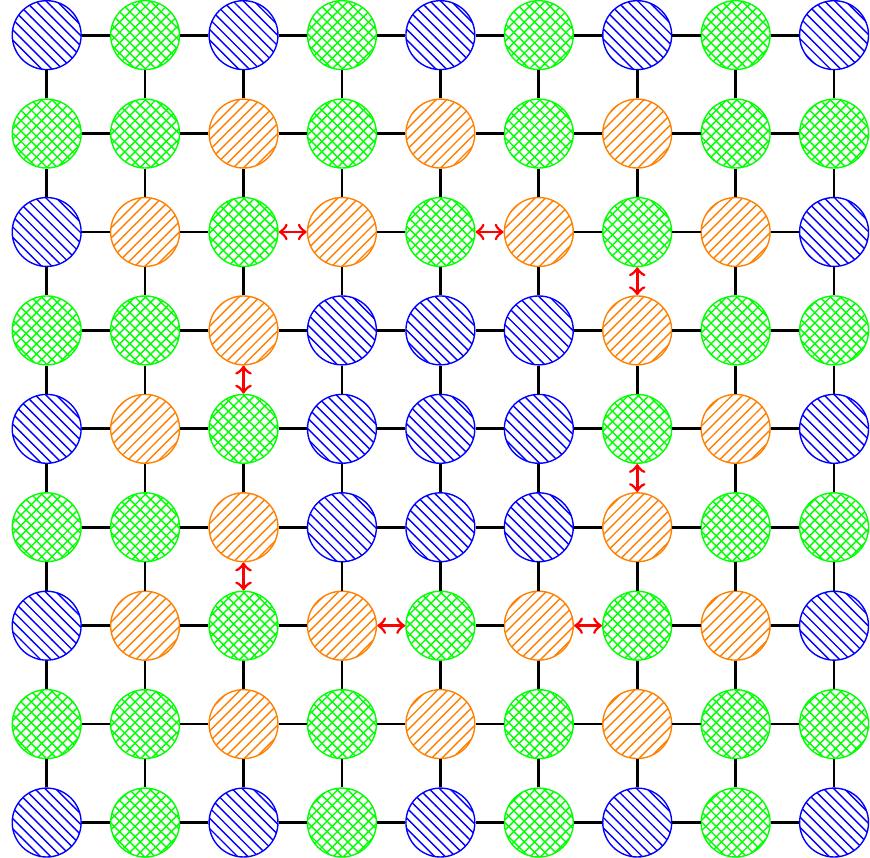}}
      \subfloat[][]{
        \label{fig:concc-9x9-5}
        \includegraphics{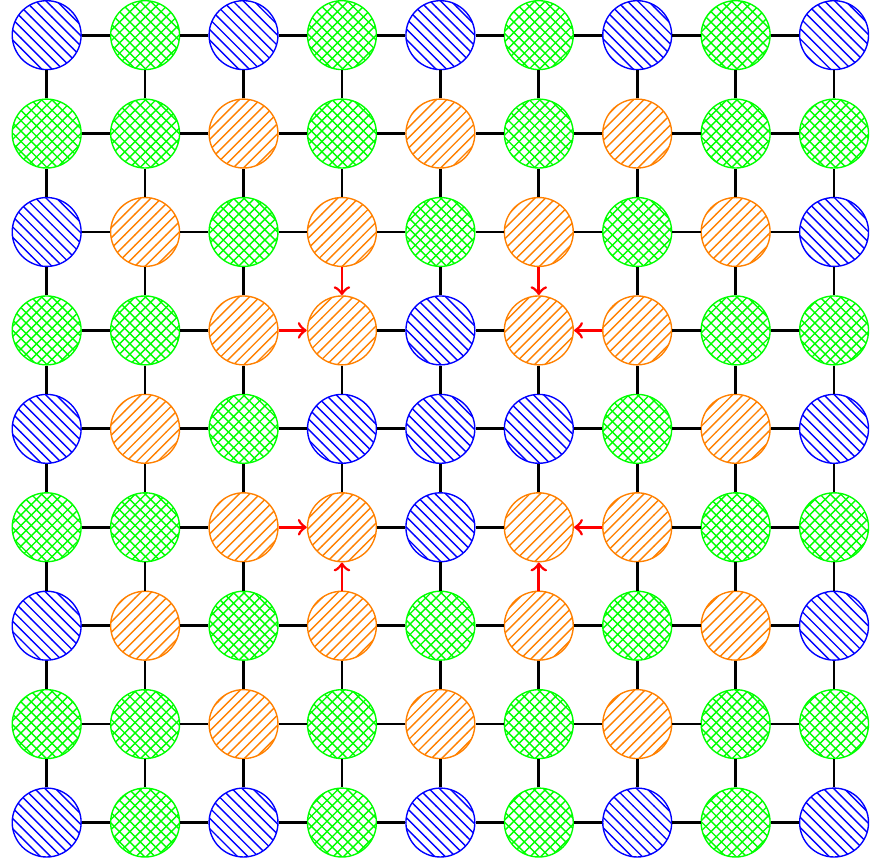}}} \\
    \noindent\makebox[\textwidth]{
      \subfloat[][]{
        \label{fig:concc-9x9-6}
        \includegraphics{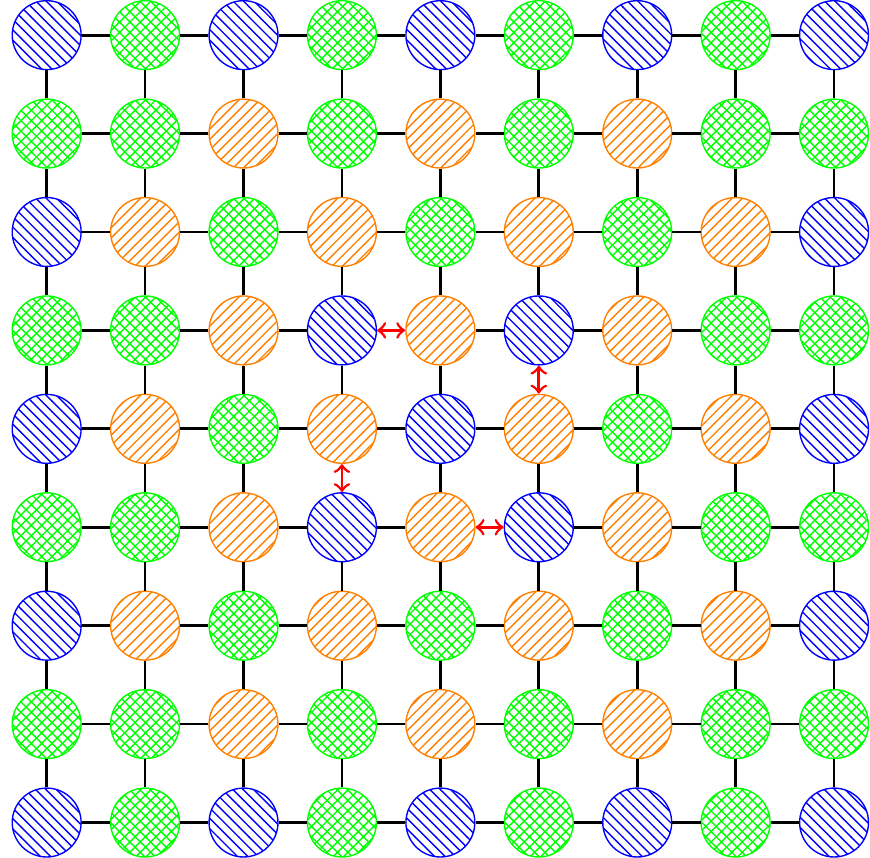}}
      \subfloat[][]{
        \label{fig:concc-9x9-7}
        \includegraphics{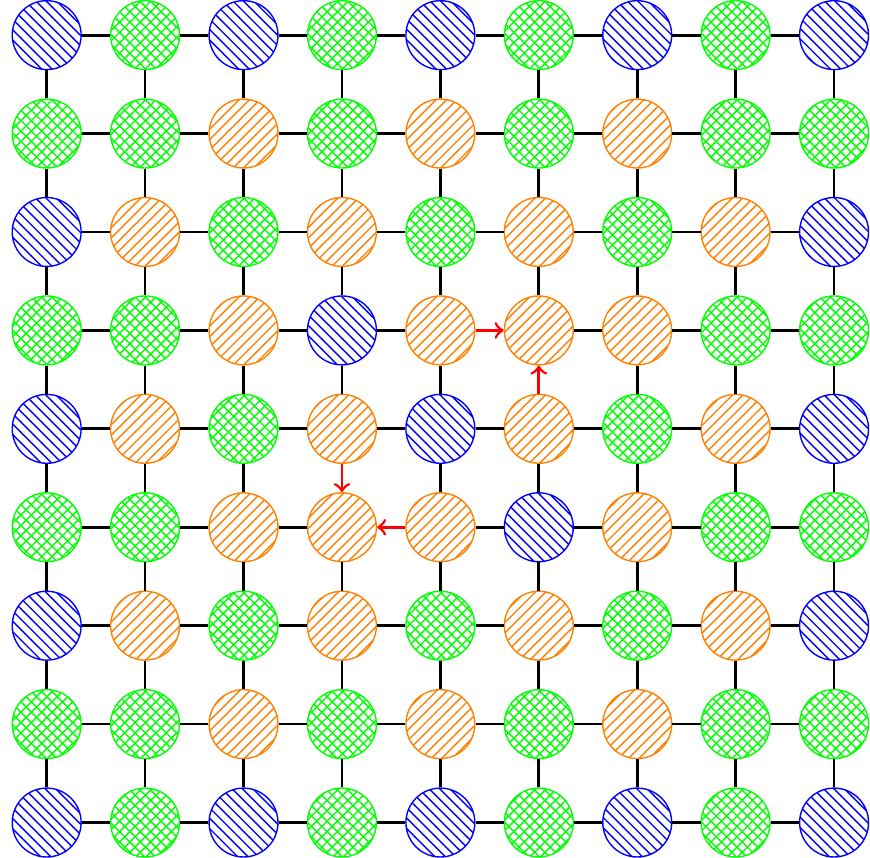}}}
    \caption[]{A controlled operation on a $9 \times 9$ grid}
  \end{figure}

  \begin{figure}[H]
    \ContinuedFloat
    \centering
    \noindent\makebox[\textwidth]{
      \subfloat[][]{
        \label{fig:concc-9x9-8}
        \includegraphics{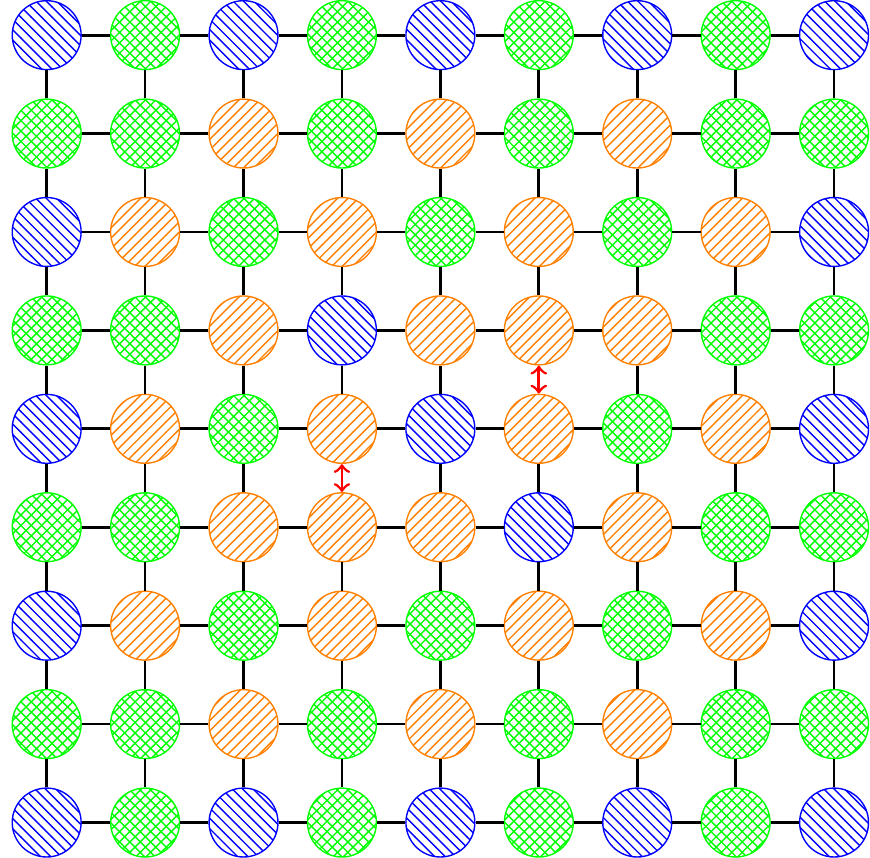}}
      \subfloat[][]{
        \label{fig:concc-9x9-9}
        \includegraphics{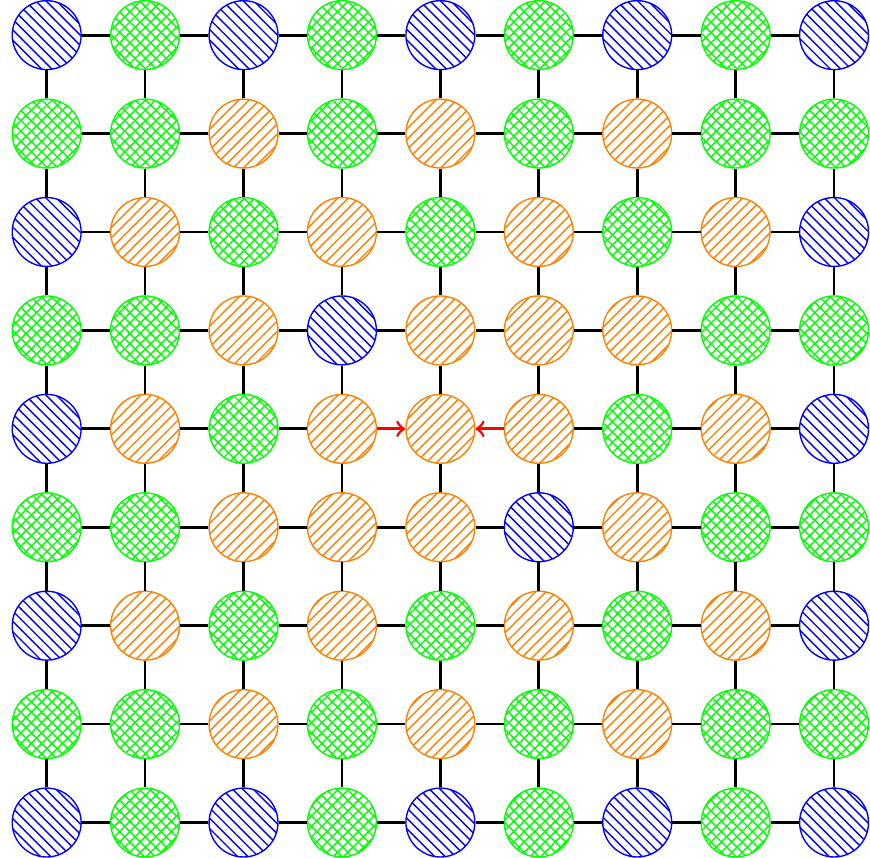}}}
    \caption[]{A controlled operation on a $9 \times 9$ grid}
  \end{figure}

  \resetsubfig}

\notintqc{\references}

\end{document}